\documentclass[final,journal,twocolumn]{IEEEtran}
\ifCLASSINFOpdf
\else
\fi
\usepackage{amsthm, amssymb}
\usepackage{graphicx}
\usepackage{bm}
\usepackage{mathtools}
\usepackage{enumitem}
\usepackage{url}
\usepackage{cite}

\theoremstyle{definition}
\newtheorem{dfn}{Definition}
\theoremstyle{plain}
\newtheorem{thm}{Theorem}
\theoremstyle{plain}
\newtheorem{cor}{Corollary}
\theoremstyle{plain}
\newtheorem{lem}{Lemma}
\theoremstyle{definition}
\newtheorem{ex}{Example}
\theoremstyle{definition}
\newtheorem*{rem}{Remark}

\DeclareMathOperator*{\argmax}{arg\,max}
\DeclareMathOperator*{\argsup}{arg\,sup}

\begin{document}
%
% paper title
% Titles are generally capitalized except for words such as a, an, and, as,
% at, but, by, for, in, nor, of, on, or, the, to and up, which are usually
% not capitalized unless they are the first or last word of the title.
% Linebreaks \\ can be used within to get better formatting as desired.
% Do not put math or special symbols in the title.
\title{Mismatched Decoding Reliability\\Function at Zero Rate}
%
%
% author names and IEEE memberships
% note positions of commas and nonbreaking spaces ( ~ ) LaTeX will not break
% a structure at a ~ so this keeps an author's name from being broken across
% two lines.
% use \thanks{} to gain access to the first footnote area
% a separate \thanks must be used for each paragraph as LaTeX2e's \thanks
% was not built to handle multiple paragraphs
%

\author{Marco~Bondaschi,~\IEEEmembership{Graduate~Student~Member,~IEEE,}
        Albert~Guill\'en~i~F\`abregas,~\IEEEmembership{Fellow,~IEEE,}\\
        and~Marco~Dalai,~\IEEEmembership{Senior~Member,~IEEE}% <-this % stops a space
\thanks{M. Bondaschi is with the School of Computer and Communication Sciences, \'{E}cole Polytechnique F\'{e}d\'{e}rale de Lausanne, CH-1015 Lausanne, Switzerland (e-mail: marco.bondaschi@epfl.ch).}%
\thanks{A. Guill\'en i F\`abregas is with the Department of Engineering, University of Cambridge, Cambridge CB2 1PZ, U.K., and with the Department of Information and Communication Technologies, Universitat Pompeu Fabra, Barcelona 08018, Spain (e-mail: guillen@ieee.org).}% <-this % stops a space
\thanks{M. Dalai is with the Department of Information Engineering at the University of Brescia, Via Branze 38 I-25123 Brescia, Italy (e-mail: marco.dalai@unibs.it).}%
\thanks{This work was supported in part by the European Research Council under Grant 725411.}%
\thanks{This research was partially supported by Italian Ministry of Education under Grant PRIN 2015 D72F16000790001.}%
\thanks{This work was presented in part at the 2021 IEEE International Symposium on Information Theory, Melbourne, Australia, Jul. 2021.}% <-this % stops a space
\thanks{Copyright (c) 2021 IEEE. Personal use of this material is permitted.  However, permission to use this material for any other purposes must be obtained from the IEEE by sending a request to pubs-permissions@ieee.org.}%
}

\maketitle

% As a general rule, do not put math, special symbols or citations
% in the abstract or keywords.
\begin{abstract}
We derive an upper bound on the reliability function of mismatched decoding for zero-rate codes. The bound is based on a result by Koml\'os that shows the existence of a subcode with certain symmetry properties. The bound is shown to coincide with the expurgated exponent at rate zero for a broad family of channel-decoding metric pairs.
\end{abstract}

% Note that keywords are not normally used for peerreview papers.

% For peer review papers, you can put extra information on the cover
% page as needed:
% \ifCLASSOPTIONpeerreview
% \begin{center} \bfseries EDICS Category: 3-BBND \end{center}
% \fi
%
% For peerreview papers, this IEEEtran command inserts a page break and
% creates the second title. It will be ignored for other modes.
\IEEEpeerreviewmaketitle

\section{Introduction}
% The very first letter is a 2 line initial drop letter followed
% by the rest of the first word in caps.
% 
% form to use if the first word consists of a single letter:
% \IEEEPARstart{A}{demo} file is ....
% 
% form to use if you need the single drop letter followed by
% normal text (unknown if ever used by the IEEE):
% \IEEEPARstart{A}{}demo file is ....
% 
% Some journals put the first two words in caps:
% \IEEEPARstart{T}{his demo} file is ....
% 
% Here we have the typical use of a "T" for an initial drop letter
% and "HIS" in caps to complete the first word.
Consider a discrete memoryless channel with finite input alphabet $\mathcal{X}$ and output alphabet $\mathcal{Y}$, and with transition probabilities $W(y|x)$. For a message set $\mathcal{M}=\{1,2,\ldots,M\}$ and blocklength $n$, an encoder is a function $\mathcal{C} : \mathcal{M}  \to \mathcal{X}^n$ that assigns to each message $m$ a corresponding codeword $\boldsymbol{x}_m=(x_{m,1},x_{m,2},\ldots,x_{m,n})$. The \emph{rate of transmission} is defined as
\begin{equation}
R\triangleq \frac{\log M}{n}\,.
\end{equation}
When message $m$ is sent, an output sequence $\boldsymbol{y}=(y_1,y_2,\ldots, y_n)$ is received with probability 
\begin{equation}
W^n(\boldsymbol{y}|\boldsymbol{x}_m) = \prod_{i=1}^n W(y_i|x_{m,i})\,.
\end{equation}
A decoder is a function $\mathcal{C}^{-1}: \mathcal{Y}^n \to \mathcal{M}$ whose task is to map each possible output sequence to a message in $\mathcal{M}$ which hopefully is equal to the message that was originally sent with high probability. In this paper, we consider a decoder that follows the rule
\begin{equation}
\label{Decoder}
\mathcal{C}^{-1}(\boldsymbol{y}) \in \{m \in \mathcal{M} : q^n(\boldsymbol{x}_m,\boldsymbol{y}) \geq q^n(\boldsymbol{x}_{m'},\boldsymbol{y}) \,\,\forall m' \in \mathcal{M}\}\,,
\end{equation}
where
\begin{equation}
q^n(\boldsymbol{x}_m,\boldsymbol{y}) = \prod_{i=1}^n q(x_{m,i}, y_i)
\end{equation}
for a fixed function $q: \mathcal{X} \times \mathcal{Y} \to \mathbb{R}^+$ called \emph{decoding metric}. We assume for now that when there is a tie, that is, when for a certain output sequence $\boldsymbol{y}$ the maximal $q^n(\boldsymbol{x}_m,\boldsymbol{y})$ is attained by more than one message, the decoder selects one of them with an arbitrary rule. However, as we will discuss in more detail later on, most of the results in this work are valid only for a decoder that breaks ties equiprobably among the messages that maximize $q^n(\boldsymbol{x}_m, \boldsymbol{y})$. We will distinguish this case from the general one when necessary.

When $q(x,y) = W(y|x)$, the decoder is the maximum likelihood decoder, achieving the lowest probability of error. Instead, when the decoding metric $q(x,y) \neq W(y|x)$, the decoder is, in general, said to be mismatched \cite{csiszar1, merhav1}. The mismatched decoding problem encompasses a number of important problems such as channel uncertainty, fading channels, reduced-complexity decoding, bit-interleaved coded modulation, optical communications, and zero-error and zero-undetected error capacity. See \cite{scarlett1} for a recent survey of the information theoretic foundations of mismatched decoding.

When message $m$ is sent, a decoding error occurs if $\mathcal{C}^{-1}(\boldsymbol{y}) \neq m$. The probability of this event is
\begin{equation}
\label{PemDef}
P_{e,m}^{(n)} \triangleq \sum_{\boldsymbol{y} \in \mathcal{Y}^c_m} W^n(\boldsymbol{y}|\boldsymbol{x}_m)\,,
\end{equation}
where $\mathcal{Y}_m \subset \mathcal{Y}^n$ is the subset of output sequences that are decoded to $m$. The average probability of error of the code is 
\begin{equation}
P_{e}^{(n)} \triangleq \frac{1}{M} \sum_{m=1}^M P_{e,m}^{(n)}\,.
\end{equation}

For fixed $R$, $n$ and decoding metric $q$, let $P_e^q(R,n)$ be the smallest probability of error for mismatched decoding over all codes with rate at least $R$ and block length $n$, when $q$ is used as the decoding metric. The mismatched reliability function is defined as 
\begin{equation}
E^q(R) \triangleq \limsup_{n \to \infty} -\frac{\log P_e^q(R,n)}{n}
\end{equation}
and represents the asymptotic exponent with which the probability of error goes to zero for a given channel and decoding metric, when an optimal code with blocklength $n$ and rate at least $R$ is used. The supremum of the information rates $R$ for which the error probability tends to zero is called \emph{mismatched capacity}.

In general, there is no single-letter expression for the mismatched capacity. A number of achievable rate and error exponent results based on random coding are available in the literature \cite{csiszar1, merhav1, fischer1, kaplan1, somekh1}. In terms of upper bounds on the mismatched capacity or on the reliability function, there are fewer results in the literature. Recently, single-letter upper bounds on the mismatched capacity improving over the Shannon capacity were proposed in \cite{kangarshahi1,somekh2}. A sphere-packing upper bound on the mismatched reliability function was recently derived in \cite{kangarshahi2}, yielding an improved upper bound on the mismatched capacity.

In this paper, we study the problem of finding an upper bound on the mismatched reliability function of any given discrete memoryless channel and decoding metric, when the rate tends to $0$, that is, we are interested in upper-bounding $E^q(0^+)$.
In the following, we  focus only on decoding metrics that are meaningful for our problem. In particular, we restrict our attention to decoding metrics such that
\begin{equation}
\label{MetricCondition1}
W(y|x) > 0 \implies q(x,y)>0
\end{equation}
for all $x \in \mathcal{X}$ and $y \in \mathcal{Y}$. In fact, channels with a decoding metric that does not meet this condition for some input $x$ have a mismatched capacity equal to $0$ if that input is used \cite{csiszar1}, and so they are of little interest.

As we already mentioned, several lower bounds on the mismatched reliability function exist. The one that is most relevant for this work is a generalization of Gallager's classical expurgated bound to the case of mismatched decoding obtained by Scarlett \emph{et al.} \!\cite{scarlett2}; when the rate approaches zero, their bound takes the form
\begin{align}
\label{ELower}
E^q(0^+) &\geq \notag\\
	&\hspace{-3.5em}\max_{Q \in \mathcal{P}(\mathcal{X})} \sup_{s \geq 0} -\!\!\sum_{a \in \mathcal{X}}\sum_{b \in \mathcal{X}}Q(a) Q(b)\log\sum_{y \in \mathcal{Y}} W(y|a) \left(\frac{q(b,y)}{q(a,y)}\right)^{\!\!s}.
\end{align}

In the following, we  derive an upper bound on $E^q(0^+)$ for a wide class of channels and decoding metrics, under the assumption that ties are broken equiprobably. Such an upper bound will turn out to be equal to the lower bound \eqref{ELower}; therefore, for such a class of channels, the bound \eqref{ELower} is tight.

In order to prove our bound, in Section II we study conditions that channels and decoding metrics must satisfy in order to have a finite reliability function at rate $R=0^+$. Then, in Section III we derive a lower bound on the mismatched probability of error for two codewords, and in section IV we prove the tightness of \eqref{ELower} using the lower bound of Section III and a probabilistic result by Koml\'{o}s (obtained using Ramsey theory) on the existence of a subset of random variables from a larger set that (asymptotically) have pairwise symmetric distributions. The application of these ideas in coding theory originated in works of Blinovsky, see for example \cite{blinovsky1,blinovsky2,ahlswede1}. See also \cite{bondaschi1} for a recent revisitation of the maximum likelihood case.

\section{Mismatched zero-error capacity}

In the following we assume that condition \eqref{MetricCondition1} is satisfied. It is also useful to restrict our attention only to channels and decoding metrics such that at all rates $R>0$ the minimum probability of error is strictly positive; if this is not the case, then at $R=0^+$ the reliability function is infinite and no finite upper bound is possible. Thus, we introduce a new quantity, the mismatched zero-error capacity $C^q_0$ for a channel $W(y|x)$ and a decoding metric $q(x,y)$, defined as the supremum of the rates $R$ for which there exist codes with probability of error exactly equal to zero. If $C_0^q$ is positive for some channel and decoding metric, then the reliability function at $R=0^+$ is infinite. Hence, we would like to restrict our attention to channels and decoding metrics with $C_0^q = 0$.

We now proceed to analyze more closely the conditions for a positive mismatched zero-error capacity. Notice that $C_0^q$ is positive if and only if there exist two codewords $\boldsymbol{x}_1$ and $\boldsymbol{x}_2$ (of arbitrary blocklength) such that for all output sequences $\boldsymbol{y}$:
\begin{enumerate}
\item either $W^n(\boldsymbol{y}|\boldsymbol{x}_1)=0$ or $W^n(\boldsymbol{y}|\boldsymbol{x}_2)=0$;
\item $W^n(\boldsymbol{y}|\boldsymbol{x}_1) > 0 \implies q^n(\boldsymbol{x}_1,\boldsymbol{y}) \geq q^n(\boldsymbol{x}_2,\boldsymbol{y})$\\
$W^n(\boldsymbol{y}|\boldsymbol{x}_2) > 0 \implies q^n(\boldsymbol{x}_2,\boldsymbol{y}) \geq q^n(\boldsymbol{x}_1,\boldsymbol{y}).$
\end{enumerate}
Condition 1 states that each possible output sequence can be obtained only from one of the two codewords; condition 2 states that each sequence is always decoded correctly. Notice that up to now we are still assuming that ties can be decoded with an arbitrary rule. That is why Condition 2 admits cases such that $q^n(\boldsymbol{x}_1,\boldsymbol{y}) = q^n(\boldsymbol{x}_2,\boldsymbol{y})$: there always exists a tie-breaking rule that decodes each output sequence correctly in these cases (the one that associates to each $\boldsymbol{y}$ the only $\boldsymbol{x}$ with $W^n(\boldsymbol{y}|\boldsymbol{x}) > 0$).

As we stated earlier on, our upper bound on $E^q(0^+)$ only holds for a decoder that breaks ties equiprobably. Therefore, even if the definition of mismatched zero-error capacity given above is the most general, since it admits any decoding strategy for breaking ties, it is nonetheless meaningful to us to introduce a second definition of mismatched zero-error capacity, that we denote by $\bar{C}_0^q$, that is the supremum of the rates $R$ for which there exist codes with probability of error exactly equal to zero, given that ties are broken equiprobably. Since this decoding strategy is not necessarily the best one, that is, the one that achieves the minimum probability of error, it follows that, in general, $\bar{C}_0^q \leq C_0^q$.

As for the conditions stated above, the only difference, when $\bar{C}_0^q$ is considered instead of $C_0^q$, is that condition 2 becomes:
\begin{enumerate}[label=\arabic*b)]
\stepcounter{enumi}
\item $W^n(\boldsymbol{y}|\boldsymbol{x}_1) > 0 \implies q^n(\boldsymbol{x}_1,\boldsymbol{y}) > q^n(\boldsymbol{x}_2,\boldsymbol{y})$\\
$W^n(\boldsymbol{y}|\boldsymbol{x}_2) > 0 \implies q^n(\boldsymbol{x}_2,\boldsymbol{y}) > q^n(\boldsymbol{x}_1,\boldsymbol{y})$
\end{enumerate}
since in this case ties are not allowed, given that, with ties broken equiprobably, there is always a positive probability of decoding the output sequence incorrectly.

Finally, when the chosen decoding metric is the maximum-likelihood one, that is, $q(x,y) = W(y|x)$, the two zero-error capacities introduced above are equal and coincide with the classical zero-error capacity $C_0$; also, since the maximum-likelihood decoding metric is the one that minimizes the probability of error, we have that, for any decoding metric, $\bar{C}_0^q \leq C_0^q \leq C_0$. 

The next objective of this section is to find conditions for the mismatched zero-error capacity to be zero that depend only on the single-letter channel probabilities $W(y|x)$ and decoding metric $q(x,y)$. This can be done using the same tools that we will need to study $E^q(0^+)$. Therefore, we introduce now a real-valued function that will be useful to both ends. For any two sequences $\boldsymbol{x}_1$ and $\boldsymbol{x}_2$ in $\mathcal{X}^n$, we define
\begin{equation}
\mu_{\boldsymbol{x}_1,\boldsymbol{x}_2}(s) \triangleq -\log \sum_{\boldsymbol{y}\in \hat{\mathcal{Y}}^n_{\boldsymbol{x}_1,\boldsymbol{x}_2}} W^n(\boldsymbol{y}|\boldsymbol{x}_1) \bigg(\frac{q^n(\boldsymbol{x}_2,\boldsymbol{y})}{q^n(\boldsymbol{x}_1,\boldsymbol{y})}\bigg)^s\,,\label{MuDef}
\end{equation}
where
\begin{equation}
\hat{\mathcal{Y}}^n_{\boldsymbol{x}_1,\boldsymbol{x}_2} \triangleq \big\{\boldsymbol{y} \in\mathcal{Y}^n : q^n(\boldsymbol{x}_1,\boldsymbol{y}) q^n(\boldsymbol{x}_2,\boldsymbol{y}) > 0\big\}\,.
\end{equation}
When $n=1$, \eqref{MuDef} becomes, for any $a,b \in \mathcal{X}$,
\begin{equation}
\label{MuAB}
\mu_{a,b}(s) \triangleq -\log\sum_{y \in \hat{\mathcal{Y}}_{a,b}} W(y|a) \bigg(\frac{q(b,y)}{q(a,y)}\bigg)^s\,,
\end{equation}
with
\begin{equation}
\hat{\mathcal{Y}}_{a,b} \triangleq \big\{y \in \mathcal{Y} : q(a,y) q(b,y) > 0\big\}.
\end{equation}
An additional quantity that will be useful to us is the limit of the derivative of $\mu_{\boldsymbol{x}_1,\boldsymbol{x}_2}(s)$ when $s \to \infty$; for this reason, we introduce a compact symbol for it:
\begin{gather}
\mu'_{\boldsymbol{x}_1,\boldsymbol{x}_2} \triangleq \lim_{s \to\infty} \mu'_{\boldsymbol{x}_1,\boldsymbol{x}_2}(s) \label{MuPrimeX} \\
\mu'_{a,b} \triangleq \lim_{s \to\infty} \mu'_{a,b}(s) \label{MuPrimeAB}
\end{gather}
and we set by definition $\mu'_{\boldsymbol{x}_1,\boldsymbol{x}_2} = +\infty$ if $\mu_{\boldsymbol{x}_1,\boldsymbol{x}_2}(s) = +\infty$, and the same for $\mu'_{a,b}$.\footnote{Throughout the paper we use the convention $\frac{\cdot}{0} = +\infty$.}
\begin{lem}
\label{MuEtaProperties}
The following properties hold for all $a,b \in \mathcal{X}$ and all sequences $\boldsymbol{x}_1, \boldsymbol{x}_2$ in $\mathcal{X}^n$.
\begin{enumerate}
\item Let $P_{\boldsymbol{x}_1,\boldsymbol{x}_2}$ be the joint type of $\boldsymbol{x}_1$ and $\boldsymbol{x}_2$. Then,
\begin{equation}
\frac{1}{n}\,\mu_{\boldsymbol{x}_1,\boldsymbol{x}_2}(s) = \sum_{a \in \mathcal{X}} \sum_{b \in \mathcal{X}} P_{\boldsymbol{x}_1, \boldsymbol{x}_2}(a,b) \mu_{a,b}(s)\,. \label{MuType}
\end{equation}
\item $\mu_{a,b}(s)$ and $\mu_{\boldsymbol{x}_1,\boldsymbol{x}_2}(s)$ are concave.
\item
\begin{equation}
\mu_{a,a}(s) = 0\,.
\end{equation}
\item
\begin{equation}
\mu_{a,b}' = \min_{y: W(y|a)>0} \log \frac{q(a,y)}{q(b,y)}\,. \label{MuDer}
\end{equation}
\end{enumerate}
\end{lem}
\begin{proof}
To prove property 1, notice that $\mu_{\boldsymbol{x}_1,\boldsymbol{x}_2}(s)$ is additive, in the sense that it can be rewritten coordinate-by-coordinate as
\begin{equation}
\label{MuAdditive}
\mu_{\boldsymbol{x}_1,\boldsymbol{x}_2}(s) = \sum_{c=1}^n \mu_c(s)\,,
\end{equation}
where
\begin{equation}
\mu_c(s) \triangleq -\log\sum_{y \in \hat{\mathcal{Y}}_c} W(y|x_{1,c}) \bigg(\frac{q(x_{2,c},y)}{q(x_{1,c},y)}\bigg)^s
\end{equation}
and 
\begin{equation}
\hat{\mathcal{Y}}_c \triangleq \big\{y \in \mathcal{Y} : q(x_{1,c},y) q(x_{2,c},y) > 0\big\}\,.
\end{equation}
Each term of the sum in \eqref{MuAdditive} only depends on the pair of input symbols $(x_{1,c}, x_{2,c})$. Since every pair $(a,b) \in \mathcal{X}^2$ appears in $n P_{\boldsymbol{x}_1,\boldsymbol{x}_2}(a,b)$ coordinates, grouping together the equal terms in \eqref{MuAdditive} leads to \eqref{MuType}.

Property 2 can be proved by computing the first and second derivatives of $\mu_{a,b}(s)$, that is
\begin{equation}
\label{MuFirstDer}
\mu_{a,b}'(s) = -\frac{\sum_{y\in \hat{\mathcal{Y}}_{a,b}} W(y|a) \big(\frac{q(b,y)}{q(a,y)}\big)^s \log \frac{q(b,y)}{q(a,y)}}{\sum_{\bar{y}\in \hat{\mathcal{Y}}_{a,b}} W(\bar{y}|a) \big(\frac{q(b,\bar{y})}{q(a,\bar{y})}\big)^s}
\end{equation}
and
\begin{multline}
\mu_{a,b}''(s) = -\frac{\sum_{y\in \hat{\mathcal{Y}}_{a,b}} W(y|a) \big(\frac{q(b,y)}{q(a,y)}\big)^s \big(\log \frac{q(b,y)}{q(a,y)}\big)^2}{\sum_{\bar{y}\in \hat{\mathcal{Y}}_{a,b}} W(\bar{y}|a) \big(\frac{q(b,\bar{y})}{q(a,\bar{y})}\big)^s} \\+ \big[\mu_{a,b}'(s)\big]^2.
\end{multline}
Now, for any $s \geq 0$, these two quantities can be seen respectively as the expected value and the variance (with a minus sign) of a random variable; in fact, define the following probability distribution over the set of sequences $\hat{\mathcal{Y}}_{a,b}$,
\begin{equation}
Q_s(y) \triangleq \frac{W(y|a) \big(\frac{q(b,y)}{q(a,y)}\big)^s}{\sum_{\bar{y}\in \hat{\mathcal{Y}}_{a,b}} W(\bar{y}|a) \big(\frac{q(b,\bar{y})}{q(a,\bar{y})}\big)^s}\,,
\end{equation}
and define the random variable
\begin{equation}
D(y) \triangleq -\log \frac{q(b,y)}{q(a,y)}\,.
\end{equation}
Then, we can rewrite the two derivatives as
\begin{equation}
\mu_{a,b}'(s) = \mathbb{E}_{Q_s}[D] = \sum_{y\in \hat{\mathcal{Y}}_{a,b}} Q_s(y) D(y)
\end{equation}
\begin{align}
\mu_{a,b}''(s) &= -\text{Var}_{Q_s}[D] = -\mathbb{E}_{Q_s}\big[D^2\big] + \big(\mathbb{E}_{Q_s}[D]\big)^2 \\
	&= -\sum_{y\in \hat{\mathcal{Y}}_{a,b}} Q_s(y) D^2(y) + \big[\mu_{a,b}'(s)\big]^2.
\end{align}
Since the variance of a random variable is always non-negative, it follows that $\mu_{a,b}''(s) \leq 0$ and therefore that $\mu_{a,b}(s)$ is concave. By property 1, $\mu_{\boldsymbol{x}_1,\boldsymbol{x}_2}(s)$ is also concave, since it can be rewritten as a (weighted) sum of concave functions.

Property 3 follows from the fact that, from definition \eqref{MuAB},
\begin{equation}
\mu_{a,a}(s) = -\log \sum_{y \in \hat{\mathcal{Y}}_{a,a}} W(y|a) = 0\,,
\end{equation}
where the last equality follows from \eqref{MetricCondition1}, since
\begin{equation}
\big\{y \in \mathcal{Y}: W(y|a) > 0\big\} \subset \big\{y \in \mathcal{Y}: q(a,y) > 0\big\} = \hat{\mathcal{Y}}_{a,a}\,.
\end{equation}

To prove property 4, first notice that $\mu_{a,b}(s) = + \infty$ if and only if 
\begin{equation}
\label{MuInfinity}
\big\{y \in \mathcal{Y} : W(y|a) q(b,y) > 0\big\} = \varnothing\,.
\end{equation}
In such a case, the right-hand side of \eqref{MuDer} equals $+\infty$, in accordance to what we set by definition. If instead the set on the left-hand side of \eqref{MuInfinity} is not empty, then the property follows directly by taking the limit $s \to +\infty$ of the right-hand side of \eqref{MuFirstDer}, since both numerator and denominator are dominated by the exponentials with the largest base, which is equal to
\begin{equation*}
\max_{y: W(y|a)>0} \frac{q(b,y)}{q(a,y)}\,. \qedhere
\end{equation*}
\end{proof}

We are now ready to prove the following theorem on the mismatched zero-error capacities $C_0^q$ and $\bar{C}_0^q$.

\begin{thm}
\label{C0Thm}
For any given discrete memoryless channel $W(y|x)$ and decoding metric $q(x,y)$:
\begin{enumerate}
\item $C_0^q = 0$ if and only if
\begin{equation}
\label{C0Cond1}
\min_{y : W(y|a) > 0} \frac{q(a,y)}{q(b,y)} \leq \max_{y: W(y|b)>0} \frac{q(a,y)}{q(b,y)}
\end{equation}
for all $a,b \in \mathcal{X}$, and for all $a,b \in \mathcal{X}$ such that
\begin{equation}
\min_{y : W(y|a) > 0} \frac{q(a,y)}{q(b,y)} = \max_{y: W(y|b)>0} \frac{q(a,y)}{q(b,y)}
\end{equation}
there exists some $y \in \mathcal{Y}$ such that
\begin{equation}
\label{C0Cond2}
W(y|a) W(y|b) > 0\,.
\end{equation}
\item $\bar{C}_0^q = 0$ if and only if
\begin{equation}
\label{C0BarCond}
\min_{y : W(y|a) > 0} \frac{q(a,y)}{q(b,y)} \leq \max_{y: W(y|b)>0} \frac{q(a,y)}{q(b,y)}
\end{equation}
for all $a,b \in \mathcal{X}$.
\end{enumerate}
\end{thm}
\begin{cor}
\begin{equation}
\bar{C}_0^q = 0 \implies \max_{Q \in \mathcal{P}(\mathcal{X})} \sup_{s \geq 0} \sum_a \sum_b Q(a) Q(b) \mu_{a,b}(s) < + \infty\,.
\end{equation}
\end{cor}
\begin{proof}
We first show that the quantities defined in \eqref{MuDef} and \eqref{MuPrimeX} --- and consequently also \eqref{MuAB} and \eqref{MuPrimeAB} --- satisfy the following properties, for every $\boldsymbol{x}_1$ and $\boldsymbol{x}_2$:
\begin{align}
\lim_{s \to +\infty} \mu_{\boldsymbol{x}_1,\boldsymbol{x}_2}(s) = +\infty &\iff \mu_{\boldsymbol{x}_1,\boldsymbol{x}_2}' > 0 \label{MuPrimeProperty1} \\
\lim_{s \to +\infty} \mu_{\boldsymbol{x}_1,\boldsymbol{x}_2}(s) \in [0,+\infty) &\iff \mu_{\boldsymbol{x}_1,\boldsymbol{x}_2}' = 0 \label{MuPrimeProperty2} \\
\lim_{s \to +\infty} \mu_{\boldsymbol{x}_1,\boldsymbol{x}_2}(s) = -\infty &\iff \mu_{\boldsymbol{x}_1,\boldsymbol{x}_2}' < 0\,. \label{MuPrimeProperty3}
\end{align}
In fact, consider the function
\begin{align}
f_{\boldsymbol{x}_1,\boldsymbol{x}_2}(s) &\triangleq e^{-\mu_{\boldsymbol{x}_1,\boldsymbol{x}_2}(s)} \notag\\
	&= \sum_{\boldsymbol{y} \in \hat{\mathcal{Y}}^n_{\boldsymbol{x}_1,\boldsymbol{x}_2}} W^n(\boldsymbol{y}|\boldsymbol{x}_1) \bigg(\frac{q^n(\boldsymbol{x}_2,\boldsymbol{y})}{q^n(\boldsymbol{x}_1,\boldsymbol{y})}\bigg)^s. \label{fAB}
\end{align}
Since $f_{\boldsymbol{x}_1,\boldsymbol{x}_2}(s)$ is the sum of non-negative quantities, when $s \to \infty$ only three alternatives are possible: $f_{\boldsymbol{x}_1,\boldsymbol{x}_2}(s)$ tends to infinity, $f_{\boldsymbol{x}_1,\boldsymbol{x}_2}(s)$ tends to a finite positive number, or $f_{\boldsymbol{x}_1,\boldsymbol{x}_2}(s)$ tends to zero. In the first case, $f_{\boldsymbol{x}_1,\boldsymbol{x}_2}(s) \to \infty$ (and consequently $\mu_{\boldsymbol{x}_1,\boldsymbol{x}_2}(s) \to -\infty$) if and only if at least one non-zero term of the sum goes to infinity, which in turn happens if and only if
\begin{equation}
\max_{\boldsymbol{y}: W^n(\boldsymbol{y}|\boldsymbol{x}_1)>0} \log \frac{q^n(\boldsymbol{x}_2,\boldsymbol{y})}{q^n(\boldsymbol{x}_1,\boldsymbol{y})} = - \mu'_{\boldsymbol{x}_1,\boldsymbol{x}_2} > 0\,.
\end{equation}
In the second case, $f_{\boldsymbol{x}_1,\boldsymbol{x}_2}(s)$ tends to a finite positive real number if and only if at least one term of the sum tends to a finite positive number and all the other terms tend to zero, which happens if and only if 
\begin{equation}
\max_{\boldsymbol{y}: W^n(\boldsymbol{y}|\boldsymbol{x}_1)>0} \log \frac{q^n(\boldsymbol{x}_2,\boldsymbol{y})}{q^n(\boldsymbol{x}_1,\boldsymbol{y})} = - \mu'_{\boldsymbol{x}_1,\boldsymbol{x}_2} = 0\,;
\end{equation}
in such a case, the limit is
\begin{equation}
\lim_{s \to \infty} f_{\boldsymbol{x}_1,\boldsymbol{x}_2}(s) = \sum_{\mathrm{some}\ \boldsymbol{y}} W^n(\boldsymbol{y}|\boldsymbol{x}_1)\,,
\end{equation}
which is strictly positive and at most $1$, consequently
\begin{equation}
\lim_{s \to \infty} \mu_{\boldsymbol{x}_1,\boldsymbol{x}_2}(s) = -\log \sum_{\mathrm{some}\ \boldsymbol{y}} W^n(\boldsymbol{y}|\boldsymbol{x}_1)
\end{equation}
is finite and greater than or equal to $0$.
Finally, $f_{\boldsymbol{x}_1,\boldsymbol{x}_2}(s)$ tends to zero (and consequently $\mu_{\boldsymbol{x}_1,\boldsymbol{x}_2}(s) \to \infty$) if and only if all terms of the sum tend to zero, which happens if and only if
\begin{equation}
\max_{\boldsymbol{y}: W^n(\boldsymbol{y}|\boldsymbol{x}_1)>0} \log \frac{q^n(\boldsymbol{x}_2,\boldsymbol{y})}{q^n(\boldsymbol{x}_1,\boldsymbol{y})} = - \mu'_{\boldsymbol{x}_1,\boldsymbol{x}_2} < 0\,.
\end{equation}
Notice that the same properties hold also for $\mu_{a,b}(s)$ and $\mu'_{a,b}$ for any $a,b\in\mathcal{X}$, since one can choose $\boldsymbol{x}_1=a$ and $\boldsymbol{x}_2 = b$.

Next, we analyze more closely the properties that a pair of codewords $\boldsymbol{x}_1$ and $\boldsymbol{x}_2$ must have in order to satisfy conditions 1 and 2 above for a positive $C_0^q$. Condition 1 is satisfied if and only if in at least a coordinate of the pair of codewords, there is a pair of input symbols $(a,b)$ such that $W(y|a) W(y|b) = 0$ for all $y$, that is, the joint type $P_{\boldsymbol{x}_1,\boldsymbol{x}_2}$ of the two codewords must have $P_{\boldsymbol{x}_1,\boldsymbol{x}_2}(a,b) > 0$ for that pair of input symbols. This condition can be satisfied only if there actually exists a pair of symbols $(a,b)$ such that $W(y|a)W(y|b)=0$ for all $y$. Thus, a precondition for $C_0^q > 0$ is that
\begin{equation}
\mathcal{A} \triangleq \big\{(a,b) \in \mathcal{X}^2 : W(y|a)W(y|b) = 0 \text{ for all } y \in \mathcal{Y}\big\} \neq \varnothing.
\end{equation}
Notice that this is also the condition for the classical $C_0$ to be positive, which is of course a necessary condition to have $C_0^q > 0$, since, as we already pointed out, we have in general $C_0^q \leq C_0$.

Instead, thanks to \eqref{MuDer}, condition 2 is satisfied if and only if both
\begin{equation}
\label{C0DerCond}
\mu_{\boldsymbol{x}_1,\boldsymbol{x}_2}' \geq 0 \qquad \text{and} \qquad \mu_{\boldsymbol{x}_2,\boldsymbol{x}_1}' \geq 0\,.
\end{equation}
Hence, using \eqref{MuType}, there exists a pair of codewords satisfying condition 2 if and only if there exists a joint type $P$ such that both
\begin{equation}
\label{C0TypeCond}
\sum_a \sum_b P(a,b) \mu_{a,b}' \geq 0 \quad\text{and} \quad \sum_a \sum_b P(a,b) \mu_{b,a}' \geq 0\,.
\end{equation}
Any pair of codewords with a joint type $P$ satisfying \eqref{C0TypeCond} satisfies Condition 2 for a positive $C_0^q$. Now, condition \eqref{C0TypeCond} is true if and only if
\begin{equation}
\label{MaxMinMuEta1}
\sup_{P \in \mathcal{P}(\mathcal{X}^2)} \!\!\min\Big\{\sum_a \sum_b P(a,b) \mu_{a,b}'\, ,\,\sum_a \sum_b P(a,b) \mu_{b,a}'\Big\} > 0
\end{equation}
or
\begin{equation}
\label{MaxMinMuEta2}
\min\Big\{\sum_a \sum_b P(a,b) \mu_{a,b}'\, ,\,\sum_a \sum_b P(a,b) \mu_{b,a}'\Big\} = 0
\end{equation}
for some $P \in \mathcal{P}_{\mathbb{Q}}(\mathcal{X}^2)$, where $\mathcal{P}_{\mathbb{Q}}(\mathcal{X}^2)$ denotes the set of probability vectors over $\mathcal{X}^2$ with rational entries. This supremum can be computed easily. Notice first that the minimum of two linear functions is concave. Then, since the minimum of the two functions is invariant with respect to the transformation $P(a,b) \leftrightarrow P(b,a)$, its maximum is always attained (also) by a joint distribution such that $P(a,b) = P(b,a)$ for all $a,b$. In such a case, the two functions are both equal to
\begin{equation}
\label{SymmetricMin}
\sum_{a \leq b} P(a,b) (\mu_{a,b}' + \mu_{b,a}')
\end{equation}
and this quantity is maximized when all the weight is given to the largest term. Notice also that the $P$ achieving this maximum has rational entries and so it belongs to $\mathcal{P}_{\mathbb{Q}}(\mathcal{X}^2)$. Hence, thanks to \eqref{MuDer}, conditions \eqref{MaxMinMuEta1} and \eqref{MaxMinMuEta2} become
\begin{equation}
\label{C0ABCond}
\max_{a,b} \bigg(\min_{y: W(y|a)>0} \log \frac{q(a,y)}{q(b,y)} + \min_{y: W(y|b)>0} \log \frac{q(b,y)}{q(a,y)}\bigg) \geq 0\,.
\end{equation}
Thus, if \eqref{C0ABCond} is true, then we can find at least one joint type $P$ that satisfies \eqref{C0TypeCond}, and with it a set of pairs of codewords that satisfy Condition 2 for $C_0^q>0$. However, we have no guarantees that there exists a pair of codewords in this set that satisfies also Condition 1. For this to be true, it is necessary that a pair of codewords in the set has a joint type with $P(a,b) > 0$ for some $(a,b) \in \mathcal{A}$. We now investigate this issue. If the maximum in \eqref{C0ABCond} is strictly positive, then, thanks to the fact that the argument of the $\max$ in \eqref{SymmetricMin} is linear in $P$, in the neighborhood of the joint type achieving the maximum, there exists a (symmetric) joint type $\hat{P}$ that has $\hat{P}(a,b)>0$ for a pair of symbols $(a,b) \in \mathcal{A}$, and that, when put into \eqref{SymmetricMin}, still returns a positive value. Hence, the two codewords with that joint type satisfy both conditions 1 and 2, and $C_0^q$ is positive. If, instead, the maximum in \eqref{C0ABCond} is exactly zero, then, a joint type satisfying also condition 1 exists only if one of the joint types achieving the maximum already has a positive entry corresponding to a pair of symbols $(a,b) \in \mathcal{A}$, that is, there exists a pair of symbols $(a,b)$ such that 
\begin{equation}
\min_{y: W(y|a)>0} \log \frac{q(a,y)}{q(b,y)} + \min_{y: W(y|b)>0} \log \frac{q(b,y)}{q(a,y)} = 0
\end{equation}
and for all $y$, $W(y|a)W(y|b)=0$.

To summarize, $C_0^q > 0$ if and only if $\mathcal{A} \neq \varnothing$ and either
\begin{equation}
\max_{a,b} \bigg(\min_{y: W(y|a)>0} \log \frac{q(a,y)}{q(b,y)} + \min_{y: W(y|b)>0} \log \frac{q(b,y)}{q(a,y)}\bigg) > 0
\end{equation}
or there exists a pair $(a,b) \in \mathcal{A}$ such that
\begin{equation}
\min_{y: W(y|a)>0} \log \frac{q(a,y)}{q(b,y)} + \min_{y: W(y|b)>0} \log \frac{q(b,y)}{q(a,y)} = 0\,.
\end{equation}
The complementary conditions give the first part of the theorem.

The second part is a bit more straightforward. Condition 1 remains identical; as for condition 2b, condition \eqref{C0DerCond} is replaced by
\begin{equation}
\mu_{\boldsymbol{x}_1,\boldsymbol{x}_2}' > 0 \qquad \text{and} \qquad \mu_{\boldsymbol{x}_2,\boldsymbol{x}_1}' > 0\,.
\end{equation}
Following the same steps as before, we get that $\bar{C}_0^q> 0$ if and only if $\mathcal{A} \neq \varnothing$ and
\begin{equation}
\max_{a,b} \bigg(\min_{y: W(y|a)>0} \log \frac{q(a,y)}{q(b,y)} + \min_{y: W(y|b)>0} \log \frac{q(b,y)}{q(a,y)}\bigg) > 0\,.
\end{equation}
The complementary conditions give the second part of the theorem.

Finally, regarding the corollary, the implication follows from the fact that
\begin{align}
\max_{Q \in \mathcal{P}(\mathcal{X})} &\sup_{s \geq 0} \sum_a \sum_b Q(a) Q(b) \mu_{a,b}(s) \notag \\
&\hspace{-1em}= \frac{1}{2} \max_{Q \in \mathcal{P}(\mathcal{X})} \sup_{s \geq 0} \sum_a \sum_b Q(a) Q(b) \big(\mu_{a,b}(s) + \mu_{b,a}(s)\big) \\
&\hspace{-1em}\leq \frac{1}{2} \max_{Q \in \mathcal{P}(\mathcal{X})}\sum_a \sum_b Q(a) Q(b) \sup_{s \geq 0} \big(\mu_{a,b}(s) + \mu_{b,a}(s)\big)\,, \label{CorollaryIneq}
\end{align}
where the equality follows from the fact that
\begin{equation}
\sum_a \sum_b Q(a) Q(b) \mu_{a,b}(s) = \sum_a \sum_b Q(a) Q(b) \mu_{b,a}(s)\,.
\end{equation}
The quantity in \eqref{CorollaryIneq} is finite if $\bar{C}_0^q = 0$, since inequality \eqref{C0BarCond} can be rewritten as
\begin{equation}
\mu'_{a,b} + \mu'_{b,a} \leq 0\,,
\end{equation}
which by \eqref{MuPrimeProperty1} is equivalent to
\begin{equation}
\lim_{s \to +\infty} \big(\mu_{a,b}(s) +\mu_{b,a}(s)\big) < +\infty\,,
\end{equation}
which in turn implies that
\begin{equation}
\sup_{s \geq 0} \big(\mu_{a,b}(s) +\mu_{b,a}(s)\big) < +\infty
\end{equation}
since $\mu_{a,b}(s) +\mu_{b,a}(s)$ is concave.
\end{proof}

\section{Lower bound on the probability of error}

We now proceed to derive a lower bound on the minimum probability of error of any discrete memoryless channel and mismatched metric, under the assumption that $\bar{C}_0^q = 0$ and that ties are decoded equiprobably. In order to achieve this, we first derive a lower bound on the probability of error of codes with two codewords, and then we generalize the result to codes with an arbitrary number of codewords.

Following \eqref{PemDef}, the probabilities of error for the two messages 1 and 2 satisfy
\begin{align}
P_{e,1}^{(n)} &\geq \sum_{\boldsymbol{y}\notin\mathcal{Y}_1^n} W^n(\boldsymbol{y}|\boldsymbol{x}_1) \label{Pe1Def} \\[1ex]
P_{e,2}^{(n)} &\geq \sum_{\boldsymbol{y}\notin\mathcal{Y}_2^n} W^n(\boldsymbol{y}|\boldsymbol{x}_2)\,,\label{Pe2Def}
\end{align}
where
\begin{align}
\mathcal{Y}_1^n &= \{\boldsymbol{y} \in \mathcal{Y}^n : q^n(\boldsymbol{x}_1,\boldsymbol{y}) \geq q^n(\boldsymbol{x}_2,\boldsymbol{y})\} \\[1ex]
\mathcal{Y}_2^n &= \{\boldsymbol{y} \in \mathcal{Y}^n : q^n(\boldsymbol{x}_1,\boldsymbol{y}) \leq q^n(\boldsymbol{x}_2,\boldsymbol{y})\}.
\end{align}
Notice that the lower bounds are due to the fact that we consider all sequences that are tied between the two messages as correctly decoded. Also, we can restrict our attention only to sequences such that $q^n(\boldsymbol{x}_1,\boldsymbol{y})\, q^n(\boldsymbol{x}_2,\boldsymbol{y}) > 0$, that is, we can substitute $\mathcal{Y}^n$ with the set
\begin{equation}
\hat{\mathcal{Y}}^n = \{\boldsymbol{y}\in\mathcal{Y}^n : q^n(\boldsymbol{x}_1,\boldsymbol{y})\, q^n(\boldsymbol{x}_2,\boldsymbol{y}) > 0\}\,.
\end{equation}
In fact, thanks to the condition in \eqref{MetricCondition1}, if $q^n(\boldsymbol{x}_1,\boldsymbol{y})$ and $q^n(\boldsymbol{x}_1,\boldsymbol{y})$ are both zero for some sequence $\boldsymbol{y}$, then also $W^n(\boldsymbol{y}|\boldsymbol{x}_1)$ and $W^n(\boldsymbol{y}|\boldsymbol{x}_2)$ are zero, and the sequence contributes neither to $P_{e,1}$ nor to $P_{e,2}$. If instead only one of the two is zero, for example $q^n(\boldsymbol{x}_1,\boldsymbol{y})$, then $q^n(\boldsymbol{x}_1,\boldsymbol{y}) < q^n(\boldsymbol{x}_2,\boldsymbol{y})$ and the sequence would only contribute to $P_{e,1}$; however, by \eqref{MetricCondition1} we have $W^n(\boldsymbol{y}|\boldsymbol{x}_1) = 0$, and so also its contribution to $P_{e,1}$ is zero.

We now introduce some tools from the method of types developed by Csisz\'{a}r and K\"{o}rner \cite{csiszar2}. We define the \emph{conditional type} of the sequence $\boldsymbol{y}$ given the codewords $\boldsymbol{x}_1$ and $\boldsymbol{x}_2$ for any $a,b\in\mathcal{X}$ and $y \in \mathcal{Y}$ as
\begin{equation}
V_{\boldsymbol{y}}(y|a,b) \triangleq \frac{P_{\boldsymbol{x}_1,\boldsymbol{x}_2,\boldsymbol{y}}(a,b,y)}{P_{\boldsymbol{x}_1,\boldsymbol{x}_2}(a,b)}\,,
\end{equation}
where $P_{\boldsymbol{x}_1,\boldsymbol{x}_2}$ and $P_{\boldsymbol{x}_1,\boldsymbol{x}_2,\boldsymbol{y}}$ are the joint types of the pair $(\boldsymbol{x}_1,\boldsymbol{x}_2)$ and the triple $(\boldsymbol{x}_1,\boldsymbol{x}_2,\boldsymbol{y})$ respectively. In order to lighten the notation, from now on we let $P_{1,2}(a,b) = P_{\boldsymbol{x}_1,\boldsymbol{x}_2}(a,b)$. We also define the \emph{conditional Kullback-Leibler divergence} as
\begin{equation}
\label{CondKL}
D(V\| Z | P_{1,2}) \triangleq \sum_{a \in \mathcal{X}} \sum_{b \in \mathcal{X}} P_{1,2}(a,b) D\big(V(\cdot\,| a,b)\|Z(\cdot\,| a,b)\big)
\end{equation}
for any two conditional distributions $V,Z:\mathcal{X}^2 \to \mathcal{Y}$.

Now, all sequences $\boldsymbol{y}\in\mathcal{Y}^n$ with the same conditional type $V_{\boldsymbol{y}}$ also have the same values for $W^n(\boldsymbol{y}|\boldsymbol{x}_1)$, $W^n(\boldsymbol{y}|\boldsymbol{x}_2)$, $q^n(\boldsymbol{x}_1,\boldsymbol{y})$ and $q^n(\boldsymbol{x}_2,\boldsymbol{y})$, so they all give the same contribution to the probabilities of error in \eqref{Pe1Def} and \eqref{Pe2Def}. Hence, we can group them together and reformulate the probability of error as a function of conditional types instead of sequences:
\begin{align}
P_{e,1}^{(n)} &= \sum_{V\notin\mathcal{V}_1^n} W^n(V|\boldsymbol{x}_1) \label{Pe1Types} \\
P_{e,2}^{(n)} &= \sum_{V\notin\mathcal{V}_2^n} W^n(V|\boldsymbol{x}_2) \label{Pe2Types}\,,
\end{align}
where 
\begin{align}
\mathcal{V}_1^n &= \{V \in \mathcal{V}^n(\boldsymbol{x}_1,\boldsymbol{x}_2) : q^n(\boldsymbol{x}_1,V) \geq q^n(\boldsymbol{x}_2,V)\} \\
	&= \Big\{V \in \mathcal{V}^n(\boldsymbol{x}_1,\boldsymbol{x}_2) : \notag\\
	&\hspace{2em}\sum_{a,b} P_{1,2}(a,b) \sum_{y} V(y|a,b) \log \frac{q(a,y)}{q(b,y)} \geq 0\Big\}\,,
\end{align}
\begin{multline}
\mathcal{V}_2^n = \Big\{V \in \mathcal{V}^n(\boldsymbol{x}_1,\boldsymbol{x}_2) : \\
\sum_{a,b} P_{1,2}(a,b) \sum_{y} V(y|a,b) \log \frac{q(a,y)}{q(b,y)} \leq 0\Big\}\,,
\end{multline}
and $\mathcal{V}^n(\boldsymbol{x}_1,\boldsymbol{x}_2)$ is the set of all conditional types given $\boldsymbol{x}_1$ and $\boldsymbol{x}_2$.

Furthermore, if we define the two conditional distributions
\begin{align}
W_1(y|a,b) &= W(y|a) \quad\text{for all}\quad b\in\mathcal{X} \\
W_2(y|a,b) &= W(y|b) \quad\text{for all}\quad a \in\mathcal{X}
\end{align}
then from classical results of the method of types (see \cite{csiszar2}) we can derive the following lemma.
\begin{lem}
For any conditional type $V \in \mathcal{V}^n(\boldsymbol{x}_1,\boldsymbol{x}_2)$ we have
\begin{equation}
W^n(V|\boldsymbol{x}_m) \geq \frac{1}{(n+1)^{|\mathcal{X}|^2 |\mathcal{Y}|}} e^{-n D(V\|W_m|P_{1,2})}
\end{equation}
for any $m \in \mathcal{M} = \{1,2\}$.
\end{lem}
Hence, we can lower bound the probabilities of error in \eqref{Pe1Types} and \eqref{Pe2Types} as
\begin{equation}
\label{PeKL}
P_{e,m}^{(n)} \geq \sum_{V\notin \mathcal{V}^n_m} e^{-n [D(V\|W_m|P_{1,2}) + \delta_1(n)]}\,,
\end{equation}
where
\begin{equation}
\delta_1(n) = |\mathcal{X}|^2 |\mathcal{Y}|\frac{\log (n+1)}{n}\,.
\end{equation}

Also, using \eqref{MuDef}, it can be verified by substitution that
\begin{equation}
\label{MuKLRel}
\mu_{\boldsymbol{x}_1,\boldsymbol{x}_2}(s) - s \mu'_{\boldsymbol{x}_1,\boldsymbol{x}_2}(s)= nD(V_s\|W_1|P_{1,2})
\end{equation}
for the conditional distribution
\begin{equation}
\label{VHat}
V_s(y|a,b) = \frac{W(y|a)\big(\frac{q(b,y)}{q(a,y)}\big)^s}{\sum_{\bar{y}\in\hat{\mathcal{Y}}} W(\bar{y}|a)\big(\frac{q(b,\bar{y})}{q(a,\bar{y})}\big)^s}\,.
\end{equation}

We also need the following lemma about approximating a probability distribution with a type.
\begin{lem}[Shannon \cite{shannon1}]
\label{TypeApprox}
For any distribution $Q \in \mathcal{P}(\mathcal{Y})$, for any $n \in \mathbb{N}$, there exists a type $\hat{Q} \in \mathcal{T}^n(\mathcal{Y})$ such that
\begin{equation}
\lvert\hat{Q}(y) - Q(y)\rvert \leq \frac{1}{n} \quad\text{for all}\quad y\in\mathcal{Y}
\end{equation}
and $\hat{Q}(y)=0$ if $Q(y)=0$.
\end{lem}

We can now prove the following theorem on the probability of error.
\begin{thm}
\label{PeThm}
For $n$ large enough, the probabilities of error $P_{e,1}^{(n)}$ and $P_{e,2}^{(n)}$ are lower-bounded by
\begin{equation}
P_{e,1}^{(n)} \geq e^{-\mu_{\boldsymbol{x}_1,\boldsymbol{x}_2}(s) + s \mu'_{\boldsymbol{x}_1,\boldsymbol{x}_2}(s) - \delta(n)}
\end{equation}
for every $s$ such that $\mu'_{\boldsymbol{x}_1,\boldsymbol{x}_2}(s) < 0$, and
\begin{equation}
P_{e,2}^{(n)} \geq e^{-\mu_{\boldsymbol{x}_2,\boldsymbol{x}_1}(s) + s \mu'_{\boldsymbol{x}_2,\boldsymbol{x}_1}(s) - \delta(n)}
\end{equation}
for every $s$ such that $\mu'_{\boldsymbol{x}_2,\boldsymbol{x}_1}(s) < 0$, where
\begin{equation}
\delta(n) = \lvert\mathcal{X}\rvert^2 \lvert\mathcal{Y}\rvert\bigg(1 + 2\log (n+1) + \log \frac{1}{W_{\min}}\bigg)
\end{equation}
and $W_{\min} = \min_{x,y} W(y|x)$, where the minimum is over all $x \in \mathcal{X}$ and $y\in\mathcal{Y}$ such that $W(y|x) > 0$.
\end{thm}
\begin{proof}
We prove the bound for $P_{e,1}^{(n)}$; the bound for $P_{e,2}^{(n)}$ follows similarly. Notice that we can rewrite
\begin{equation}
\mu'_{\boldsymbol{x}_1,\boldsymbol{x}_2}(s) = \sum_{a,b} P_{1,2}(a,b) \sum_{y} V_s(y|a,b) \log \frac{q(a,y)}{q(b,y)}
\end{equation}
for $V_s$ as defined in \eqref{VHat}. Hence, for every $s$ such that $\mu'_{\boldsymbol{x}_1,\boldsymbol{x}_2}(s) < 0$ we have
\begin{equation}
\label{VIneq}
\sum_{a,b} P_{1,2}(a,b) \sum_{y} V_s(y|a,b) \log \frac{q(a,y)}{q(b,y)} < 0\,.
\end{equation}
Intuitively, for $n$ large enough we can approximate $V_s$ with a conditional type $\hat{V}_s$ that still satisfies \eqref{VIneq}. In fact, thanks to Lemma \ref{TypeApprox} we have
\begin{align}
\bigg\lvert\sum_{a,b} &P_{1,2}(a,b) \sum_{y} [\hat{V}_s(y|a,b)-V_s(y|a,b)]\log \frac{q(a,y)}{q(b,y)}\bigg\rvert \notag\\
	&\leq \sum_{a,b} P_{1,2}(a,b) \sum_{y} \lvert\hat{V}_s(y|a,b)-V_s(y|a,b)\rvert\bigg\lvert\!\log \frac{q(a,y)}{q(b,y)}\bigg\rvert \\
	&\leq \sum_{a,b} P_{1,2}(a,b) \sum_{y} \frac{1}{nP_{1,2}(a,b)}\bigg\lvert\!\log \frac{q(a,y)}{q(b,y)}\bigg\rvert \\
	&\leq \frac{1}{n} \sum_{a,b,y} \bigg\lvert\!\log \frac{q(a,y)}{q(b,y)}\bigg\rvert\,,
\end{align}
that goes to $0$ as $n \to \infty$.
Now, $\hat{V}_s$ does not belong to $\mathcal{V}_1^n$, so we can lower bound \eqref{PeKL} with
\begin{align}
\label{PeVBound}
P_{e,1}^{(n)} &\geq \sum_{V\notin \mathcal{V}^n_1} e^{-n [D(V\|W_1|P_{1,2}) + \delta_1(n)]} \\
	&\geq e^{-n [D(\hat{V}_s\|W_1|P_{1,2}) + \delta_1(n)]}\,. \label{PeVBound}
\end{align}
Again, since $\hat{V}_s$ and $V_s$ are close to each other, also $D(\hat{V}_s\|W_1|P_{1,2})$ and $D(V_s\|W_1|P_{1,2})$ are close to each other. As a matter of fact, we can write
\begin{align}
\lvert D(&\hat{V}_s\|W_1|P_{1,2}) - D(V_s\|W_1|P_{1,2})\rvert \notag\\
	&\leq \sum_{a,b} P_{1,2}(a,b) \sum_{y} \bigg\lvert\hat{V}_s(y|a,b)\log\frac{\hat{V}_s(y|a,b)}{W(y|a)}\notag\\
	&\hspace{9em}- V_s(y|a,b)\log\frac{V_s(y|a,b)}{W(y|a)}\bigg\rvert\,.
\end{align}
Now for each $a$, $b$ and $y$ there are two possibilities: if $V_s(y|a,b) > 0$ and $\hat{V}_s(y|a,b)=0$, then thanks to Lemma \ref{TypeApprox} we have
\begin{align}
\bigg\lvert\hat{V}_s(y|a,b)\log\frac{\hat{V}_s(y|a,b)}{W(y|a)}&- V_s(y|a,b)\log\frac{V_s(y|a,b)}{W(y|a)}\bigg\rvert \notag\\
	&\hspace{-1em}= V_s(y|a,b)\bigg\lvert\log\frac{V_s(y|a,b)}{W(y|a)}\bigg\rvert \\
	&\hspace{-1em}\leq \frac{1}{nP_{1,2}(a,b)} \bigg\lvert\log\frac{V_s(y|a,b)}{W(y|a)}\bigg\rvert\,,
\end{align}
where the term in absolute value is finite and independent of $n$. If instead both $\hat{V}_s(y|a,b)$ and $V_s(y|a,b)$ are positive, then we can apply Lagrange's mean value theorem to the function
\begin{equation}
g(x) = x \log \frac{x}{W(y|a)}\,,
\end{equation}
whose derivative is
\begin{equation}
g'(x) = \log \frac{x}{W(y|x)} + 1\,,
\end{equation}
to get
\begin{align}
\bigg\lvert\hat{V}_s(&y|a,b)\log\frac{\hat{V}_s(y|a,b)}{W(y|a)} - V_s(y|a,b)\log\frac{V_s(y|a,b)}{W(y|a)}\bigg\rvert \notag\\
	&\leq \lvert\hat{V}_s(y|a,b) - V_s(y|a,b)\rvert \bigg\lvert\log\frac{\bar{V}(y|a,b)}{W(y|a)} + 1\bigg\rvert \\
	&\leq \frac{1}{nP_{1,2}(a,b)}\big(1 + \lvert\log W(y|a)\rvert + \lvert\log \bar{V}(y|a,b)\rvert\big)
\end{align}
for some $\bar{V}(y|a,b) \in \big(V_s(y|a,b), \hat{V}_s(y|a,b)\big)$ (the interval endpoints might be inverted). 

Notice that both $\hat{V}_s(y|a,b)$ and $\bar{V}(y|a,b)$ depend implicitly on $n$. In order to make this dependence explicit, we study two cases. If $V_s(y|a,b) < \hat{V}_s(y|a,b)$, then \\$V_s(y|a,b)<\bar{V}(y|a,b)$ and
\begin{align}
\bigg\lvert\hat{V}_s&(y|a,b)\log\frac{\hat{V}_s(y|a,b)}{W(y|a)} - V_s(y|a,b)\log\frac{V_s(y|a,b)}{W(y|a)}\bigg\rvert \notag\\
	&\leq \frac{1}{nP_{1,2}(a,b)}\big(1 + \lvert\log W(y|a)\rvert + \lvert\log V_s(y|a,b)\rvert\big)\,.
\end{align}
If instead $\hat{V}_s(y|a,b) < V_s(y|a,b)$, then $\frac{1}{n} \leq \hat{V}_s(y|a,b)<\bar{V}(y|a,b)$ and
\begin{align}
\bigg\lvert\hat{V}_s(y|a,b)&\log\frac{\hat{V}_s(y|a,b)}{W(y|a)}- V_s(y|a,b)\log\frac{V_s(y|a,b)}{W(y|a)}\bigg\rvert \notag\\
	&\leq \frac{1}{nP_{1,2}(a,b)}\big(1 + \log n +\lvert\log W(y|a)\rvert \big)\,.
\end{align}
Putting it all together, for $n$ large enough we have
\begin{align}
\lvert D(\hat{V}_s\|&W_1|P_{1,2}) - D(V_s\|W_1|P_{1,2})\rvert \notag\\
	&\leq \frac{1}{n} \sum_{a,b,y} \big(1 + \log n + \lvert\log W(y|a)\rvert\big) \\
	&\leq \frac{\lvert\mathcal{X}\rvert^2 \lvert\mathcal{Y}\rvert}{n}\bigg(1 + \log n + \log \frac{1}{W_{\min}}\bigg) \triangleq \delta_2(n) \label{KLDiffBound}
\end{align}
that again goes to $0$ as $n \to \infty$. Hence, equations \eqref{MuKLRel}, \eqref{PeVBound} and \eqref{KLDiffBound} lead to
\begin{align}
P_{e,1}^{(n)} &\geq e^{-n [D(\hat{V}_s\|W_1|P_{1,2}) + \delta_1(n)]} \\
	&\geq e^{-n D(V_s\|W_1|P_{1,2}) -n\delta_1(n) -n\delta_2(n)} \\
	&\geq e^{-\mu_{\boldsymbol{x}_1,\boldsymbol{x}_2}(s) + s \mu'_{\boldsymbol{x}_1,\boldsymbol{x}_2}(s) - \delta(n)}\,,
\end{align}
which concludes the proof.
\end{proof}
Notice that Theorem \ref{PeThm} holds for arbitrary tie-breaking rules. The following corollary, instead, holds only under the assumption that ties are broken equiprobably (or in the case where ties are always counted as errors).
\begin{cor}
If ties are broken equiprobably, then:
\begin{align}
P_{e,1}^{(n)} &\geq \exp\Big\{-\sup_{s\geq 0}\mu_{\boldsymbol{x}_1,\boldsymbol{x}_2}(s) - \delta(n)\Big\} \\
P_{e,2}^{(n)} &\geq \exp\Big\{-\sup_{s\geq 0}\mu_{\boldsymbol{x}_2,\boldsymbol{x}_1}(s) - \delta(n)\Big\}\,.
\end{align}
\end{cor}
\begin{proof}
We prove again only the bound for $P_{e,1}^{(n)}$. There are three possibilities:
\begin{enumerate}
\item $\lim_{s\to\infty} \mu_{\boldsymbol{x}_1,\boldsymbol{x}_2}(s) = +\infty$;
\item $\lim_{s\to\infty} \mu_{\boldsymbol{x}_1,\boldsymbol{x}_2}(s) \in (-\infty, +\infty)$;
\item $\lim_{s\to\infty} \mu_{\boldsymbol{x}_1,\boldsymbol{x}_2}(s) =-\infty$.
\end{enumerate}
In the first case, we have $\sup_{s\geq 0}\mu_{\boldsymbol{x}_1,\boldsymbol{x}_2}(s) = +\infty$ and the bound simply becomes $P_{e,1}^{(n)} \geq 0$, which is trivial. 

In the second case, due to the concavity of $\mu_{\boldsymbol{x}_1,\boldsymbol{x}_2}(s)$, we have 
\begin{equation}
\sup_{s\geq 0}\mu_{\boldsymbol{x}_1,\boldsymbol{x}_2}(s) = \lim_{s\to\infty} \mu_{\boldsymbol{x}_1,\boldsymbol{x}_2}(s)\,.
\end{equation}
Since the limit is a finite real number, then from the definition of $\mu_{\boldsymbol{x}_1,\boldsymbol{x}_2}(s)$ in \eqref{MuDef} we can deduce that for all sequences $\boldsymbol{y}\in\hat{\mathcal{Y}}^n$ such that $W^n(\boldsymbol{y}|\boldsymbol{x}_1)>0$, that is, for all sequences $\boldsymbol{y}$ that can possibly contribute to $P_{e,1}^{(n)}$, we must have $\frac{q^n(\boldsymbol{x}_2,\boldsymbol{y})}{q^n(\boldsymbol{x}_1,\boldsymbol{y})} \leq 1$, or equivalently, $q^n(\boldsymbol{x}_2,\boldsymbol{y}) \leq q^n(\boldsymbol{x}_1,\boldsymbol{y})$. Since all sequences such that $q^n(\boldsymbol{x}_2,\boldsymbol{y}) < q^n(\boldsymbol{x}_1,\boldsymbol{y})$ do not contribute to $P_{e,1}^{(n)}$, this means that all sequences that appear in the sum \eqref{Pe1Def} are those that satisfy $q^n(\boldsymbol{x}_1,\boldsymbol{y}) = q^n(\boldsymbol{x}_2,\boldsymbol{y})$. Hence, in this case we can write
\begin{equation}
P_{e,1}^{(n)} = \sum_{\boldsymbol{y} \in \hat{\mathcal{Y}}^n_{\rm t}} W^n(\boldsymbol{y}|\boldsymbol{x}_1)\,,
\end{equation}
where
\begin{equation}
\hat{\mathcal{Y}}^n_{\rm t} = \{\boldsymbol{y}\in\hat{\mathcal{Y}}^n : q^n(\boldsymbol{x}_1,\boldsymbol{y}) = q^n(\boldsymbol{x}_2,\boldsymbol{y})\}\,.
\end{equation}
But in this particular case we also have
\begin{align}
\lim_{s\to\infty} \mu_{\boldsymbol{x}_1,\boldsymbol{x}_2}(s) &= \lim_{s\to\infty} -\log \sum_{\boldsymbol{y} \in \hat{\mathcal{Y}}^n} W^n(\boldsymbol{y}|\boldsymbol{x}_1) \bigg(\frac{q^n(\boldsymbol{x}_2,\boldsymbol{y})}{q^n(\boldsymbol{x}_1,\boldsymbol{y})}\bigg)^s \\
	&= -\log \sum_{\boldsymbol{y} \in \hat{\mathcal{Y}}^n_{\rm t}} W^n(\boldsymbol{y}|\boldsymbol{x}_1) = -\log P_{e,1}^{(n)}\,,
\end{align}
or equivalently, 
\begin{equation}
P_{e,1}^{(n)} = \exp\Big\{-\lim_{s\to\infty} \mu_{\boldsymbol{x}_1,\boldsymbol{x}_2}(s)\Big\} =  \exp\Big\{-\sup_{s\geq 0} \mu_{\boldsymbol{x}_1,\boldsymbol{x}_2}(s)\Big\}\,.
\end{equation}

In the third case, let $\hat{s} = \argmax_{s \in \mathbb{R}} \mu_{\boldsymbol{x}_1,\boldsymbol{x}_2}(s)$. If $\hat{s} \geq 0$, then thanks to the continuity of $\mu_{\boldsymbol{x}_1,\boldsymbol{x}_2}(s)$ and its derivative, we can apply Theorem \ref{PeThm} for $s \to \hat{s}$, so that $\mu_{\boldsymbol{x}_1,\boldsymbol{x}_2}(s) \to \mu_{\boldsymbol{x}_1,\boldsymbol{x}_2}(\hat{s}) = \sup_{s \geq 0} \mu_{\boldsymbol{x}_1,\boldsymbol{x}_2}(s)$ and $\mu'_{\boldsymbol{x}_1,\boldsymbol{x}_2}(s) \to 0$. If instead $\hat{s} < 0$, we can apply Theorem \ref{PeThm} with $s=0$.
\end{proof}

Corollary 2 leads to the fact that
\begin{equation}
P_e^{(n)} \geq \exp\big\{-n D_{\boldsymbol{x}_1,\boldsymbol{x}_2}^{(n)} + o(n)\big\}\,,
\end{equation}
where
\begin{equation}
\label{DMuEta}
D_{\boldsymbol{x}_1,\boldsymbol{x}_2}^{(n)} = \min \bigg\{\sup_{s\geq 0} \frac{1}{n}\,\mu_{\boldsymbol{x}_1,\boldsymbol{x}_2}(s), \sup_{s\geq 0} \frac{1}{n}\,\mu_{\boldsymbol{x}_2,\boldsymbol{x}_1}(s)\bigg\}\,.
\end{equation}

Finally, notice that if we consider a code with more than two codewords, say $M$, then there is one message $m$ such that
\begin{equation}
P_{e,m}^{(n)} \geq \exp \big\{-n D_{\min}(\mathcal{C}) + o(n)\big\}\,,
\end{equation}
where
\begin{equation}
\label{DMinDef}
D_{\min}(\mathcal{C}) \triangleq \min_{m \neq m' \in \mathcal{C}} D_{\boldsymbol{x}_m,\boldsymbol{x}_{m'}}^{(n)}\,,
\end{equation}
and therefore, for the whole code, the average probability of error is lower-bounded by
\begin{equation}
\label{PeDMinBound}
P_e^{(n)} \geq \frac{P_{e,m}^{(n)}}{M} \geq \exp \big\{-n \big(D_{\min}(\mathcal{C}) + R+o(1)\big)\big\}\,.
\end{equation}

\section{Upper bound on the reliability function}

Equation \eqref{PeDMinBound} shows that the problem of upper-bounding the exponent of the probability of error reduces to upper-bounding $D_{\min}(\mathcal{C})$. Specifically, for any rate $R > 0$, the number of codewords $M$ of every code of rate $R$ goes to infinity when the blocklength $n$ goes to infinity; hence, if we can find an upper bound on $D_{\min}(\mathcal{C})$ which is valid for all codes whose size $M$ tends to infinity, then thanks to \eqref{PeDMinBound}, that bound will also be a valid bound on
\begin{equation}
E(0^+) = \lim_{R \to 0} \limsup_{n \to \infty} -\frac{\log P_e(R,n)}{n}\,.
\end{equation}

Before going into the formal details, we give a brief outline of the proof of our bound and the intuition behind it. First of all, the minimum distance $D_{\min}(\mathcal{C})$ of any code $\mathcal{C}$ can be upper-bounded by the minimum distance of any subcode extracted from $\mathcal{C}$. Furthermore, the minimum distance $D_{\min}(\mathcal{C})$ is upper-bounded by the average distance $D_{\boldsymbol{x}_m,\boldsymbol{x}_{m'}}^{(n)}$ over all pairs of codewords in $\mathcal{C}$. Therefore, one natural way to upper bound the minimum distance of a code is to upper bound the average distance over a carefully selected subcode, that is,
\begin{align}
D_{\min}(\mathcal{C}) &\leq D_{\min}(\mathcal{\hat{C}}) \leq \frac{1}{\hat{M}(\hat{M}-1)}\sum_{m \neq m'\in\hat{\mathcal{C}}} \!\!D_{\boldsymbol{x}_m,\boldsymbol{x}_{m'}}^{(n)} \\
	&\hspace{-3em}= \frac{1}{\hat{M}(\hat{M}-1)}\notag\\
	&\hspace{-1.5em}\sum_{m \neq m'\in\hat{\mathcal{C}}} \min \bigg\{\sup_{s\geq 0} \frac{1}{n}\,\mu_{\boldsymbol{x}_m,\boldsymbol{x}_{m'}}(s), \sup_{s\geq 0} \frac{1}{n}\,\mu_{\boldsymbol{x}_{m'},\boldsymbol{x}_m}(s)\bigg\} \label{DBoundAverage}
\end{align}
for any $\hat{\mathcal{C}}\subset \mathcal{C}$ with $|\hat{\mathcal{C}}|=\hat{M}$.

The choice of the subcode $\hat{\mathcal{C}}$ is crucial, since, in general, the average in \eqref{DBoundAverage} may be too difficult to evaluate, for two reasons: for two generic codewords $\boldsymbol{x}_m$ and $\boldsymbol{x}_{m'}$, the functions $\mu_{\boldsymbol{x}_m,\boldsymbol{x}_{m'}}(s)$ and $\mu_{\boldsymbol{x}_{m'},\boldsymbol{x}_m}(s)$ can be very different from each other, and also, different pairs of codewords have in general very different values of $s$ at which the functions $\mu(s)$ attain their supremum. Luckily, we are able to overcome both these difficulties thanks to the following result, which is essentially by Koml\'{o}s \cite{komlos1}, and that was first employed in a coding setting by Blinovsky \cite{blinovsky1}, in the maximum likelihood case. We first state the following fundamental lemma. In order to lighten the notation, from now on, in all subscripts, a pair of codewords $\boldsymbol{x}_m, \boldsymbol{x}_{m'}$ will be denoted just by $m,m'$.
\begin{lem}[Koml\'{o}s \cite{komlos1}]
Consider a code $\mathcal{C}$ with $M$ codewords. If for each pair $(a,b) \in \mathcal{X}^2$ there exists a number $r_{a,b}$ such that for all $m < m'$,
\begin{equation}
\big\lvert\,P_{m,m'}(a,b) - r_{a,b}\,\big\rvert \leq \delta\,,
\end{equation}
then, for all $m \neq m'$ and $(a,b) \in \mathcal{X}^2$,
\begin{equation}
\big\lvert\,P_{m,m'}(a,b) - P_{m,m'}(b,a)\,\big\rvert \leq \frac{6}{\sqrt{M}} + 4\sqrt{\delta} + 4\delta\,.
\end{equation}
\end{lem}

Then, using Ramsey's theorem on the edge coloring of graphs (see for example \cite{diestel1}), the following result can be proved.

\begin{thm}
\label{KomlosThm}
For any positive integers $t$ and $\hat{M}$, there exists a positive integer $M_0(\hat{M},t)$ such that from any code $\mathcal{C}$ with $M>M_0(\hat{M},t)$ codewords, a subcode $\hat{\mathcal{C}} \subset \mathcal{C}$ with $\hat{M}$ codewords can be extracted such that for any $m\neq m'$ and $\bar{m} \neq \bar{m}'$ (not necessarily different from $m$ and $m'$) in $\hat{\mathcal{C}}$, and any $(a,b) \in \mathcal{X}^2$,
\begin{equation}
\big\lvert\,P_{m,m'}(a,b) - P_{\bar{m},\bar{m}'}(a,b)\,\big\rvert \leq \Delta(\hat{M},t)\,,
\end{equation}
where
\begin{equation}
\label{DeltaDef}
\Delta(\hat{M},t) \triangleq \frac{6}{\sqrt{\hat{M}}} + 2\sqrt{\frac{2}{t}} + \frac{3}{t}\,.
\end{equation}
\end{thm}

A proof of the two previous results in the more general list-decoding setting can be found in \cite{bondaschi1}.
Koml\'{o}s' result shows that for any positive integer $\hat{M}$, any code with an appropriately large number of codewords contains a subcode of size $\hat{M}$, whose codewords satisfy certain symmetry properties, namely: 
\begin{enumerate}[label=(\roman*)]
\item all pairs of codewords have approximately the same joint type;
\item the joint types are also approximately symmetrical, that is, $P(a,b) \simeq P(b,a)$ for all $a,b$.
\end{enumerate}
Thanks to property \eqref{MuType}, the fact that all pairs of codewords have similar joint types implies that they also have similar $\mu(s)$, while the fact that these types are close to symmetrical implies that $\mu_{\boldsymbol{x}_m,\boldsymbol{x}_{m'}}(s)$ and $\mu_{\boldsymbol{x}_{m'},\boldsymbol{x}_m}(s)$ are close to each other. However, technical problems arise due to the presence of the suprema in \eqref{DBoundAverage}, since even if the joint types are close to each other, the suprema of the functions $\mu(s)$ might be very different if they are approached as $s \to \infty$. This constrains our study only to a (very wide) class of channels and decoding metrics for which we are sure that at least one supremum in the definition of each $D_{\boldsymbol{x}_m,\boldsymbol{x}_{m'}}^{(n)}$ is attained at an $s$ no larger than a known fixed value. The class is the following.

\begin{dfn}
A discrete memoryless channel $W(y|x)$ and a decoding metric $q(x,y)$ form a \emph{balanced pair} if $\bar{C}_0^q = 0$ and for every pair $(a,b) \in \mathcal{X}^2$ belonging to the set
\begin{equation}
\label{BSet}
\mathcal{B} \triangleq \left\{(a,b) \in \mathcal{X}^2 : \min_{y : W(y|a) > 0} \frac{q(a,y)}{q(b,y)} = \max_{y: W(y|b)>0} \frac{q(a,y)}{q(b,y)}\right\}
\end{equation}
there exists a constant $B(a,b)$ such that
\begin{equation}
\label{BalancedCond}
\frac{q(a,y)}{q(b,y)} = B(a,b)
\end{equation}
for all $y \in \hat{\mathcal{Y}}_{a,b}$ such that $W(y|a) + W(y|b) > 0$.
\end{dfn}
Notice that all channels and decoding metrics such that $\bar{C}_0^q = 0$ and
\begin{equation}
\label{Subclass}
W(y|x) > 0 \iff q(x,y) > 0
\end{equation}
are balanced pairs, and indeed represent a very important special case. To see this, consider a channel-metric pair satisfying \eqref{Subclass}; for any $(a,b) \in \mathcal{B}$, we can partition the set of possible output symbols in $\hat{\mathcal{Y}}_{a,b}$ into three subsets:
\begin{align}
\mathcal{S}_a &= \{y: W(y|a) > 0 \quad\text{and}\quad W(y|b)=0\} \\
\mathcal{S}_b &= \{y: W(y|a) = 0 \quad\text{and}\quad W(y|b)>0\} \\
\mathcal{S}_{ab} &= \{y: W(y|a) > 0 \quad\text{and}\quad W(y|b)>0\}\,.
\end{align}
For all $y \in \mathcal{S}_a$, $q(a,y) > 0$ and $q(b,y) = 0$, therefore $q(a,y)/q(b,y) = +\infty$. Similarly, $q(a,y)/q(b,y) = 0$ for all $y \in \mathcal{S}_b$. Hence, we have that
\begin{align}
\min_{y : W(y|a) > 0} \frac{q(a,y)}{q(b,y)} &= \min_{y : W(y|a)W(y|b) > 0} \frac{q(a,y)}{q(b,y)} \\
\max_{y: W(y|b)>0} \frac{q(a,y)}{q(b,y)} &= \max_{y: W(y|a)W(y|b)>0} \frac{q(a,y)}{q(b,y)}\,,
\end{align}
and since $(a,b) \in \mathcal{B}$, these two quantities must be equal, that is,
\begin{equation}
\min_{y : W(y|a)W(y|b) > 0} \frac{q(a,y)}{q(b,y)} = \max_{y: W(y|a)W(y|b)>0} \frac{q(a,y)}{q(b,y)}\,,
\end{equation}
which means that the ratio $q(a,y)/q(b,y)$ must be equal for all possible $y \in \hat{\mathcal{Y}}_{a,b}$. This proves that the channel-metric pair is indeed balanced.
Furthermore, for this particular subclass,
\begin{equation}
C_0 = 0 \iff C_0^q = 0 \iff \bar{C}_0^q = 0\,,
\end{equation}
where $C_0$ is the classical zero-error capacity.
An example of a non-balanced channel-decoding metric pair is the following.
\begin{ex}
Consider the three-input typewriter channel with $\mathcal{X} = \mathcal{Y} = \{0,1,2\}$ and crossover probabilities $W(1|0) = W(2|1) = W(0|2) = \varepsilon$, with $0 < \varepsilon < 2 - \sqrt{2}$. Furthermore, consider a decoding metric such that $q(a,y) = W(y|a)$ for all $a$ and $y$ with the exception of $q(1,0) = q(1,2) = \frac{\varepsilon}{2}$ (Fig. 1). 

\begin{figure}
\centering
\includegraphics[draft=false,width=0.25\textwidth]{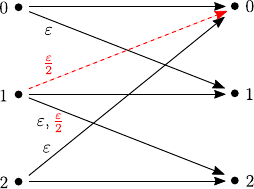}
\caption{Unbalanced channel-decoding metric pair of Example 1.}
\end{figure}

As one can check, this channel-decoding metric satisfies the condition for $\bar{C}_0^q = 0$ in Theorem 1, since for all $a,b \in \mathcal{X}$,
\begin{equation}
\min_{y: W(y|a) > 0} \frac{q(a,y)}{q(b,y)} = \max_{y: W(y|b) > 0} \frac{q(a,y)}{q(b,y)}\,.
\end{equation}
Notice that for this channel also the classical zero-error capacity $C_0$ is zero.
However, this channel-decoding metric pair does not satisfy the second condition in Definition 1 for a balanced pair, since
\begin{equation}
\frac{q(0,0)}{q(1,0)} = \frac{2(1-\varepsilon)}{\varepsilon} \neq \frac{q(0,1)}{q(1,1)} = \frac{\varepsilon}{1-\varepsilon}\,.
\end{equation}
\end{ex}

One last lemma that will be useful in bounding the average in \eqref{DBoundAverage} is the following, which is a standard trick employed, for example, in the derivation of the Plotkin bound.

\begin{lem}
\label{PlotkinLemma}
For any code with $\hat{M}$ codewords of blocklength $n$, for any $a, b \in \mathcal{X}$, with $a \neq b$,
\begin{equation}
\sum_{m \neq m'} P_{m,m'}(a,b) = \frac{1}{n} \sum_{c=1}^n \hat{M}_c(a) \hat{M}_c(b)\,,
\end{equation}
where $\hat{M}_c(a)$ is the number of times the symbol $a$ occurs in the coordinate $c$ in all the codewords.
\end{lem}
\begin{proof}
Imagine the code as an $\hat{M} \times n$ matrix, with each codeword as a row. Then, $\sum_{m\neq m'} n P_{m,m'}(a,b)$ is the number of times the pair $(a,b)$ can be found by selecting any two entries of the matrix belonging to the same column. The same computation can be performed column by column: for a generic column $c$, that number is simply the number of times $a$ occurs in that column, multiplied by the number of times $b$ occurs. Finally, summing over all columns returns the same number of the first computation, thus proving the lemma.
\end{proof}

We are ready to proceed and prove the upper bound on the reliability function at $R=0^+$ for any balanced channel-metric pair. In particular, we will show that for this class, the $D_{m,m'}^{(n)}$ are all close to each other for all pairs of codewords in the symmetric subcode $\hat{\mathcal{C}} \subset \mathcal{C}$, whose existence is guaranteed by Theorem 3. In order to show this, first of all, for any concave function $f(s)$, let\footnote{Here $f(+\infty)$ means $\lim_{s \to +\infty} f(s)$. If $\lim_{s \to +\infty} f(s) = +\infty$, then $\mathcal{S} = \{+\infty\}$, since $f(s)$ is concave.}
\begin{equation}
\mathcal{S} \triangleq \Big\{0 \leq s \leq +\infty : f(s) = \sup_{s \geq 0} f(s)\Big\}
\end{equation}
and define
\begin{equation}
\argsup_{s \geq 0} f(s) \triangleq \inf \mathcal{S}\,.
\end{equation}
In the following lemma we prove that in the case of balanced pairs, for all pairs of codewords, the concave functions $\mu_{m,m'}(s) + \mu_{m',m}(s)$ achieve their suprema at an $s$ in a bounded interval, determined only by the channel and the decoding metric.
\begin{lem}
\label{SHatLemma}
For any balanced pair, for any pair of codewords $m,m'$,
\begin{equation}
\label{MuEtaArgSup}
\argsup_{s \geq 0} \big(\mu_{m,m'}(s) + \mu_{m',m}(s)\big) \in [\,0,\,\hat{s}\,]\,,
\end{equation}
where
\begin{equation}
\label{SHat}
\hat{s} \triangleq \max_{a,b}\Big\{ \argsup_{s \geq 0} \big(\mu_{a,b}(s) + \mu_{b,a}(s)\big)\Big\} < +\infty\,.
\end{equation}
\end{lem}
\begin{proof}
We first show that $\hat{s}$ is finite. We already pointed out in the proof of Theorem \ref{C0Thm}, that equation \eqref{C0BarCond} can be rewritten as
\begin{equation}
\mu'_{a,b} + \mu'_{b,a} \leq 0\,.
\end{equation}
For the $(a,b) \in \mathcal{X}^2$ such that $\mu'_{a,b} + \mu'_{b,a} < 0$, we have that
\begin{equation}
\lim_{s \to +\infty} \mu_{a,b}(s) + \mu_{b,a}(s) = -\infty\,,
\end{equation}
which in turn implies that there exists a finite $\hat{s}_{a,b} \geq 0$ such that
\begin{equation}
\mu_{a,b}(\hat{s}_{a,b}) + \mu_{b,a}(\hat{s}_{a,b}) = \max_{s \geq 0} \big(\mu_{a,b}(s) + \mu_{b,a}(s)\big)
\end{equation}
since $\mu_{a,b}(s) + \mu_{b,a}(s)$ is concave. For the $(a,b) \in \mathcal{X}^2$ such that $\mu'_{a,b} + \mu'_{b,a} = 0$, instead, equation \eqref{BalancedCond} implies that
\begin{align}
\mu_{a,b}(s) &= -\log\sum_{y \in \hat{\mathcal{Y}}_{a,b}} W(y|a) \bigg(\frac{q(b,y)}{q(a,y)}\bigg)^s \\
	&= -\log\sum_{y \in \hat{\mathcal{Y}}_{a,b}} W(y|a) B(a,b)^{-s} \\
	&= s \log B(a,b) -\log\sum_{y \in \hat{\mathcal{Y}}_{a,b}} W(y|a)\,,
\end{align}
which is a straight line. Furthermore, since ${B(b,a) = 1/B(a,b)}$, we have that
\begin{align}
\mu_{a,b}&(s) + \mu_{b,a}(s) \notag\\
	&=  s \log B(a,b)B(b,a) -\log\sum_{y \in \hat{\mathcal{Y}}_{a,b}} W(y|a) \notag\\
	&\hspace{11em}-\log\sum_{y \in \hat{\mathcal{Y}}_{a,b}} W(y|b) \\
	&= -\log\sum_{y \in \hat{\mathcal{Y}}_{a,b}} W(y|a) -\log\sum_{y \in \hat{\mathcal{Y}}_{a,b}} W(y|b)\,,
\end{align}
which is a constant. Hence, 
\begin{equation}
\sup_{s \geq 0} \big(\mu_{a,b}(s) + \mu_{b,a}(s)\big) = \mu_{a,b}(0) + \mu_{b,a}(0)
\end{equation}
and we can set $\hat{s}_{a,b} =0$ for these $(a,b)$. Thus, the $\hat{s}$ defined by \eqref{SHat} can be rewritten as
\begin{equation}
\hat{s} = \max_{a,b} \,\hat{s}_{a,b}\,,
\end{equation}
which is finite, since all $\hat{s}_{a,b}$ are finite.

Equation \eqref{MuEtaArgSup} follows from the fact that
\begin{multline}
\mu'_{m,m'}(\hat{s}) + \mu'_{m',m}(\hat{s}) \\
= n \sum_a \sum_b P_{m,m'}(a,b) \big(\mu'_{a,b}(\hat{s}) + \mu'_{b,a}(\hat{s})\big) \leq 0\,,
\end{multline}
where the equality is due to \eqref{MuType}, while the inequality is due to the fact that for all $(a,b)$,
\begin{equation}
\mu'_{a,b}(\hat{s}) + \mu'_{b,a}(\hat{s}) \leq \mu'_{a,b}(\hat{s}_{a,b}) + \mu'_{b,a}(\hat{s}_{a,b}) \leq 0\,,
\end{equation}
since $\hat{s}_{a,b} \leq \hat{s}$ and $\mu_{a,b}(s) + \mu_{b,a}(s)$ is concave.
\end{proof}

The previous lemma and the symmetry properties of the codewords in $\hat{\mathcal{C}}$ lead to the following fundamental result, that shows that the $D_{m,m'}^{(n)}$ are close to each other for all pairs of codewords in $\hat{\mathcal{C}}$. This fact is what will make the computation of the average in \eqref{DBoundAverage} possible.
\begin{lem}
\label{SBarLemma}
For any balanced pair, for any pair of codewords $m,m' \in \hat{\mathcal{C}}$, let
\begin{equation}
\label{SBarDef}
\bar{s}_{m,m'} \triangleq \min \Big\{\argsup_{s \geq 0}\mu_{m,m'}(s), \argsup_{s \geq 0} \mu_{m',m}(s)\Big\}\,.
\end{equation}
Then, $0 \leq \bar{s}_{m,m'} \leq \hat{s}$, with $\hat{s}$ defined by \eqref{SHat}, and
\begin{equation}
\label{DMuBound}
D_{m,m'}^{(n)} \leq \frac{1}{n}\,\mu_{m,m'}(\bar{s}_{m,m'}) + K \Delta(\hat{M},t)
\end{equation}
with $\Delta(\hat{M},t)$ as defined by \eqref{DeltaDef}, and
\begin{equation}
K \triangleq \max_{0 \leq s \leq \hat{s}} \sum_a \sum_b \big\lvert\mu_{a,b}(s)\big\rvert\,.
\end{equation}
Furthermore, for any other pair of codewords $\bar{m}, \bar{m}' \in \hat{\mathcal{C}}$,
\begin{equation}
\label{MuDifferenceBound}
\bigg\lvert \frac{1}{n}\,\mu_{\bar{m},\bar{m}'}(\bar{s}_{\bar{m},\bar{m}'}) - \frac{1}{n}\,\mu_{\bar{m},\bar{m}'}(\bar{s}_{m,m'})\bigg\rvert \leq 4K \Delta(\hat{M},t)\,.
\end{equation}
\end{lem}
\begin{proof}
To prove that $\bar{s}_{m,m'} \leq \hat{s}$, notice that from equation \eqref{MuEtaArgSup} we get
\begin{equation}
\mu'_{m,m'}(\hat{s}) + \mu'_{m',m}(\hat{s}) \leq 0\,,
\end{equation}
which is possible only if
\begin{equation}
\mu'_{m,m'}(\hat{s}) \leq 0 \quad \text{or} \quad \mu'_{m',m}(\hat{s}) \leq 0\,,
\end{equation}
which in turn implies that
\begin{equation}
\argsup_{s \geq 0}\mu_{m,m'}(s) \leq \hat{s} \quad \text{or} \quad \argsup_{s \geq 0} \mu_{m',m}(s) \leq \hat{s}\,.
\end{equation}
This proves that
\begin{equation}
\bar{s}_{m,m'} \triangleq \min \Big\{\argsup_{s \geq 0}\mu_{m,m'}(s), \argsup_{s \geq 0} \mu_{m',m}(s)\Big\} \leq \hat{s}\,.
\end{equation}

Next, definition \eqref{SBarDef} implies that
\begin{equation}
\label{SupAlternative1}
\sup_{s \geq 0} \mu_{m,m'}(s) = \mu_{m,m'}(\bar{s}_{m,m'})
\end{equation}
or
\begin{equation}
\label{SupAlternative2}
\sup_{s \geq 0} \mu_{m',m}(s) = \mu_{m',m}(\bar{s}_{m,m'}).
\end{equation}
Hence, we have that, thanks to \eqref{DMuEta},
\begin{equation}
D_{m,m'}^{(n)} \leq \frac{1}{n} \,\mu_{m,m'}(\bar{s}_{m,m'}) \quad \text{or} \quad D_{m,m'}^{(n)} \leq \frac{1}{n} \,\mu_{m',m}(\bar{s}_{m,m'})\,.
\end{equation}
In the first case, equation \eqref{DMuBound} follows immediately; in the second case, we have that
\begin{align}
\bigg\lvert \frac{1}{n}\,&\mu_{m',m}(\bar{s}_{m,m'}) - \frac{1}{n}\,\mu_{m,m'}(\bar{s}_{m,m'})\bigg\rvert \notag\\
	&= \Big\lvert \sum_a \sum_b \big(P_{m,m'}(a,b) - P_{m,m'}(b,a)\big)\mu_{a,b}(\bar{s}_{m,m'})\Big\rvert \\
	&\leq \sum_a \sum_b \big\lvert P_{m,m'}(a,b) - P_{m',m}(a,b)\big\rvert \big\lvert\mu_{a,b}(\bar{s}_{m,m'})\big\rvert \\
	&\leq \Delta(\hat{M},t) \sum_a \sum_b \big\lvert\mu_{a,b}(\bar{s}_{m,m'})\big\rvert \\
	&\leq \Delta(\hat{M},t) \max_{0 \leq s \leq \hat{s}} \sum_a \sum_b \big\lvert\mu_{a,b}(s)\big\rvert \\
	&= K \Delta(\hat{M},t) \label{MuEtaDifferenceBound}
\end{align}
and therefore,
\begin{equation}
D_{m,m'}^{(n)} \leq \frac{1}{n} \,\mu_{m',m}(\bar{s}_{m,m'}) \leq \frac{1}{n}\,\mu_{m,m'}(\bar{s}_{m,m'}) + K \Delta(\hat{M},t)\,.
\end{equation}

Finally, in order to prove \eqref{MuDifferenceBound}, first notice that
\begin{align}
\bigg\lvert \frac{1}{n}\,&\mu_{\bar{m},\bar{m}'}(\bar{s}_{\bar{m},\bar{m}'}) - \frac{1}{n}\,\mu_{\bar{m},\bar{m}'}(\bar{s}_{m,m'})\bigg\rvert \notag\\
	&= \frac{1}{n}\big\lvert \mu_{\bar{m},\bar{m}'}(\bar{s}_{\bar{m},\bar{m}'}) - \mu_{m,m'}(\bar{s}_{m,m'}) \notag\\
	&\hspace{6em}+ \mu_{m,m'}(\bar{s}_{m,m'}) - \mu_{\bar{m},\bar{m}'}(\bar{s}_{m,m'})\big\rvert \\
	&\leq \frac{1}{n}\big\lvert \mu_{\bar{m},\bar{m}'}(\bar{s}_{\bar{m},\bar{m}'}) - \mu_{m,m'}(\bar{s}_{m,m'})\big\rvert \notag\\
	&\hspace{4.5em}+ \frac{1}{n}\big\lvert \mu_{m,m'}(\bar{s}_{m,m'}) - \mu_{\bar{m},\bar{m}'}(\bar{s}_{m,m'})\big\rvert\,. \label{TwoAbsBound}
\end{align}
The second absolute value can be bounded as follows:
\begin{align}
\frac{1}{n}\big\lvert &\mu_{m,m'}(\bar{s}_{m,m'}) - \mu_{\bar{m},\bar{m}'}(\bar{s}_{m,m'})\big\rvert \notag\\
	&= \Big\lvert \sum_a \sum_b \big(P_{m,m'}(a,b) - P_{\bar{m},\bar{m}'}(a,b)\big)\mu_{a,b}(\bar{s}_{m,m'})\Big\rvert \\
	&\leq \sum_a \sum_b \big\lvert P_{m,m'}(a,b) - P_{\bar{m},\bar{m}'}(a,b)\big\rvert \big\lvert\mu_{a,b}(\bar{s}_{m,m'})\big\rvert \\
	&\leq \Delta(\hat{M},t) \sum_a \sum_b \big\lvert\mu_{a,b}(\bar{s}_{m,m'})\big\rvert \\
	&\leq K \Delta(\hat{M},t)\,, \label{MuMuDifferenceBound}
\end{align}
which also holds for every $0\leq s \leq \hat{s}$.
The first absolute value, instead, can be bounded in the following way. Suppose that
\begin{equation}
\frac{1}{n}\,\mu_{\bar{m},\bar{m}'}(\bar{s}_{\bar{m},\bar{m}'}) \geq \frac{1}{n}\,\mu_{m,m'}(\bar{s}_{m,m'})\,;
\end{equation}
the other case can be proved in the same way. Then, we can write that
\begin{multline}
\label{MuDifferenceCase}
\frac{1}{n}\big\lvert \mu_{\bar{m},\bar{m}'}(\bar{s}_{\bar{m},\bar{m}'}) - \mu_{m,m'}(\bar{s}_{m,m'})\big\rvert \\= \frac{1}{n}\,\mu_{\bar{m},\bar{m}'}(\bar{s}_{\bar{m},\bar{m}'}) - \frac{1}{n}\,\mu_{m,m'}(\bar{s}_{m,m'})\,.
\end{multline}
Furthermore, thanks to \eqref{SupAlternative1} and \eqref{SupAlternative2}, we have two alternatives. If 
\begin{equation*}
\sup_{s \geq 0} \mu_{m,m'}(s) = \mu_{m,m'}(\bar{s}_{m,m'})\,,
\end{equation*}
then
\begin{equation}
\frac{1}{n}\,\mu_{m,m'}(\bar{s}_{m,m'}) \geq \frac{1}{n}\,\mu_{m,m'}(\bar{s}_{\bar{m},\bar{m}'})
\end{equation}
and we can bound \eqref{MuDifferenceCase} by
\begin{align}
\frac{1}{n}\,\mu_{\bar{m},\bar{m}'}&(\bar{s}_{\bar{m},\bar{m}'}) - \frac{1}{n}\,\mu_{m,m'}(\bar{s}_{m,m'}) \notag\\
	&\leq \frac{1}{n}\,\mu_{\bar{m},\bar{m}'}(\bar{s}_{\bar{m},\bar{m}'}) - \frac{1}{n}\,\mu_{m,m'}(\bar{s}_{\bar{m},\bar{m}'}) \\
	&\leq K \Delta(\hat{M},t)
\end{align}
as in \eqref{MuMuDifferenceBound}. If instead
\begin{equation*}
\sup_{s \geq 0} \mu_{m',m}(s) = \mu_{m',m}(\bar{s}_{m,m'})\,,
\end{equation*}
then
\begin{align}
\frac{1}{n}\,\mu_{m,m'}(\bar{s}_{m,m'}) &\geq \frac{1}{n}\,\mu_{m',m}(\bar{s}_{m,m'}) - K\Delta(\hat{M},t) \\
	&\geq \frac{1}{n}\,\mu_{m',m}(\bar{s}_{\bar{m},\bar{m}'}) - K\Delta(\hat{M},t) \\
	&\geq \frac{1}{n}\,\mu_{m,m'}(\bar{s}_{\bar{m},\bar{m}'}) - 2K\Delta(\hat{M},t)
\end{align}
using \eqref{MuEtaDifferenceBound} twice. Hence, we can bound \eqref{MuDifferenceCase} by
\begin{align}
\frac{1}{n}\,&\mu_{\bar{m},\bar{m}'}(\bar{s}_{\bar{m},\bar{m}'}) - \frac{1}{n}\,\mu_{m,m'}(\bar{s}_{m,m'}) \notag\\
	&\leq \frac{1}{n}\,\mu_{\bar{m},\bar{m}'}(\bar{s}_{\bar{m},\bar{m}'}) - \frac{1}{n}\,\mu_{m,m'}(\bar{s}_{\bar{m},\bar{m}'}) +2K\Delta(\hat{M},t)\\
	&\leq 3K \Delta(\hat{M},t)\,,
\end{align}
again as in \eqref{MuMuDifferenceBound}. Putting this and \eqref{MuMuDifferenceBound} into \eqref{TwoAbsBound} leads to \eqref{MuDifferenceBound}.
\end{proof}

Finally, thanks to this lemma, we can prove our upper bound on the reliability function at $R=0^+$, which coincides with the lower bound \eqref{ELower}.
\begin{thm}
\label{EBoundThm}
For any balanced pair,
\begin{align}
E^q(0^+) =& \notag\\
&\hspace{-4em}\max_{Q \in \mathcal{P}(\mathcal{X})} \sup_{s \geq 0} -\!\!\sum_{a \in \mathcal{X}}\sum_{b \in \mathcal{X}}Q(a) Q(b)\log\!\!\sum_{y \in \hat{\mathcal{Y}}_{a,b}} \!\!W(y|a) \biggl(\frac{q(b,y)}{q(a,y)}\biggr)^{\!\!s}.
\end{align}
\end{thm}
\begin{proof}
We already pointed out that for any subcode of $\mathcal{C}$, and in particular for the subcode $\hat{\mathcal{C}}$ of Theorem \ref{KomlosThm}, we have
\begin{equation}
\label{DBoundAverage2}
D_{\min}(\mathcal{C}) \leq D_{\min}(\hat{\mathcal{C}}) \leq \frac{1}{\hat{M}(\hat{M}-1)} \sum_{m \neq m'} D_{m,m'}^{(n)}
\end{equation}
with $m,m' \in \hat{\mathcal{C}}$. Then, we can bound the average as follows, similarly as what Shannon, Gallager and Berlekamp did in the maximum likelihood setting for the particular class of pairwise reversible channels \cite{sgb1}. Fix any pair of codewords $\hat{m} \neq \hat{m}' \in \hat{\mathcal{C}}$. Then, 

\begin{align}
D_{\min}&(\hat{\mathcal{C}}) \notag\\
	&\hspace{-1.5em}\leq \frac{1}{\hat{M}(\hat{M}-1)} \sum_{m \neq m'} D_{m,m'}^{(n)} \\
	&\hspace{-1.5em}\leq \frac{1}{\hat{M}(\hat{M}-1)} \sum_{m \neq m'} \frac{1}{n}\,\mu_{m,m'}(\bar{s}_{m,m'}) +K\Delta(\hat{M},t) \label{DMinUpperBound1}\\
	&\hspace{-1.5em}\leq \frac{1}{\hat{M}(\hat{M}-1)} \sum_{m \neq m'} \frac{1}{n}\,\mu_{m,m'}(\bar{s}_{\hat{m},\hat{m}'}) +5K\Delta(\hat{M},t) \label{DMinUpperBound2} \\
	&\hspace{-1.5em}= \frac{1}{\hat{M}(\hat{M}-1)} \sum_a \sum_b \sum_{m\neq m'} P_{m,m'}(a,b) \mu_{a,b}(\bar{s}_{\hat{m},\hat{m}'}) \notag\\
		&\hspace{13.5em}+5K\Delta(\hat{M},t)\label{DMinUpperBound25} \\
	&\hspace{-1.5em}= \frac{1}{n} \frac{1}{\hat{M}(\hat{M}-1)} \sum_{c=1}^n \sum_a \sum_b \hat{M}_c(a) \hat{M}_c(b) \mu_{a,b}(\bar{s}_{\hat{m},\hat{m}'}) \notag\\
		&\hspace{13.5em}+5K\Delta(\hat{M},t) \label{DMinUpperBound3}\\
	&\hspace{-1.5em}= \frac{1}{n} \frac{\hat{M}}{\hat{M}-1} \sum_{c=1}^n \sum_a \sum_b \frac{\hat{M}_c(a)}{\hat{M}} \frac{\hat{M}_c(b)}{\hat{M}} \mu_{a,b}(\bar{s}_{\hat{m},\hat{m}'}) \notag\\
		&\hspace{13.5em}+5K\Delta(\hat{M},t)
\end{align}
\begin{align}
	&\leq \frac{\hat{M}}{\hat{M}-1} \max_{Q \in \mathcal{P}(\mathcal{X})} \sum_a \sum_b Q(a) Q(b) \mu_{a,b}(\bar{s}_{\hat{m},\hat{m}'}) \notag\\
		&\hspace{15em}+5K\Delta(\hat{M},t) \label{DMinUpperBound4}\\
	&\leq \frac{\hat{M}}{\hat{M}-1} \sup_{s \geq 0} \max_{Q \in \mathcal{P}(\mathcal{X})} \sum_a \sum_b Q(a) Q(b) \mu_{a,b}(s) \notag\\
		&\hspace{15em}+5K\Delta(\hat{M},t)\,, \label{DMinUpperBoundLast}
\end{align}
where \eqref{DMinUpperBound1} is due to \eqref{DMuBound}, \eqref{DMinUpperBound2} is due to \eqref{MuDifferenceBound}, \eqref{DMinUpperBound3} is due to Lemma \ref{PlotkinLemma}, and \eqref{DMinUpperBound4} is due to the fact that for every $c$, 
\begin{equation*}
\bigg\{\frac{\hat{M}_c(a)}{\hat{M}}, \quad a \in \mathcal{X}\bigg\}
\end{equation*}
is a probability distribution over $\mathcal{X}$. As we already underlined, these steps are possible thanks to the fact that all pairs of codewords in $\hat{\mathcal{C}}$ have joint types that are both symmetrical and close to each other, and that this combined with the fact that for all balanced pairs we can focus the attention only on the $s$ in a known bounded interval, all the $D_{m,m'}^{(n)}$ that appear in the average \eqref{DBoundAverage2} are close to each other.
Then, letting $M \to \infty$ (so that we may also let $\hat{M} \to \infty$, by Theorem \ref{KomlosThm}) and $t \to \infty$ we obtain, thanks to the fact that $\Delta(\hat{M},t) \to 0$ in \eqref{DMinUpperBoundLast},
\begin{equation}
D_{\min}(\mathcal{C}) \leq \sup_{s \geq 0} \max_{Q \in \mathcal{P}(\mathcal{X})} \sum_a \sum_b Q(a) Q(b) \mu_{a,b}(s)\,,
\end{equation}
which is independent of the code $\mathcal{C}$.
Finally, thanks to equation \eqref{PeDMinBound}, since we let $R \to 0$ \emph{after} $n \to \infty$, we obtain an upper bound on the reliability function at $R=0^+$:
\begin{align}
E^q(0^+) &\leq\notag\\
&\hspace{-3.5em} \max_{Q \in \mathcal{P}(\mathcal{X})} \sup_{s \geq 0}-\sum_a \sum_b Q(a) Q(b) \log\!\!\sum_{y \in \hat{\mathcal{Y}}_{a,b}} \!\!W(y|a) \biggl(\frac{q(b,y)}{q(a,y)}\biggr)^{\!\!s},
\end{align}
which equals the expurgated lower bound given by \eqref{ELower}, proving the theorem.
\end{proof}

If the channel and decoding metric are not a balanced pair, our method fails in that for some pair $(a,b) \in \mathcal{X}^2$ belonging to the set $\mathcal{B}$ defined in \eqref{BSet}, the function $\mu_{a,b}(s) + \mu_{b,a}(s)$ is concave and has a horizontal asymptote at $s \to +\infty$, but it is not a straight line; because of this, a finite $\hat{s}$ as in Lemma \ref{SHatLemma} cannot be determined. A partial solution to this problem is to upper-bound these functions by their horizontal asymptote. This strategy leads to a similar upper bound as the one for balanced pairs; however, in this case, the bound is larger than the expurgated bound at $R=0^+$. To obtain this bound, define for any pair $(a,b) \in \mathcal{B}$,
\begin{equation}
A(a,b) \triangleq \min_{y: W(y|a)>0} \frac{q(a,y)}{q(b,y)} = \max_{y: W(y|b) > 0} \frac{q(a,y)}{q(b,y)}\,,
\end{equation}
and let
\begin{equation}
\label{YASet}
\hat{\mathcal{Y}}_{a,b}^A \triangleq \bigg\{y \in \hat{\mathcal{Y}}_{a,b} : \frac{q(a,y)}{q(b,y)} = A(a,b)\bigg\}\,.
\end{equation}
Then, we can upper-bound $\mu_{a,b}(s)$ and $\mu_{b,a}(s)$ as follows.
\begin{align}
\mu_{a,b}(s) &\triangleq -\log \sum_{y \in \hat{\mathcal{Y}}_{a,b}} W(y|a) \bigg(\frac{q(b,y)}{q(a,y)}\bigg)^s \\
	&\leq -\log \sum_{y \in \hat{\mathcal{Y}}_{a,b}^A} W(y|a) A(a,b)^{-s} \\
	&= s \log A(a,b) - \log \sum_{y \in \hat{\mathcal{Y}}_{a,b}^A} W(y|a)
\end{align}
and in the same way,
\begin{equation}
\mu_{b,a}(s) \leq -s \log A(a,b) - \log \sum_{y \in \hat{\mathcal{Y}}_{a,b}^A} W(y|b).
\end{equation}
Now, if we define 
\begin{gather}
\hat{\mu}_{a,b}(s) \triangleq s \log A(a,b) - \log \sum_{y \in \hat{\mathcal{Y}}_{a,b}^A} W(y|a) \\
\hat{\mu}_{b,a}(s) \triangleq -s \log A(a,b) - \log \sum_{y \in \hat{\mathcal{Y}}_{a,b}^A} W(y|b)
\end{gather}
we have $\mu_{a,b}(s) \leq \hat{\mu}_{a,b}(s)$ and $\mu_{b,a}(s) \leq \hat{\mu}_{b,a}(s)$, and
\begin{equation}
\hat{\mu}_{a,b}(s) + \hat{\mu}_{b,a}(s) = - \log \sum_{y \in \hat{\mathcal{Y}}_{a,b}^A} W(y|a) - \log \sum_{y \in \hat{\mathcal{Y}}_{a,b}^A} W(y|b)
\end{equation}
which is constant. Finally, if we set $\hat{\mu}_{a,b}(s) \triangleq \mu_{a,b}(s)$ for all pairs $(a,b) \in \mathcal{B}^c$, one can readily check that Lemma \ref{SHatLemma}, Lemma \ref{SBarLemma} and Theorem \ref{EBoundThm} (the upper bound part) still hold for any discrete memoryless channel and decoding metric if the $\mu_{a,b}(s)$ are replaced with $\hat{\mu}_{a,b}(s)$. Hence, for a generic pair of channel and decoding metric, the following theorem can be proved.

\begin{thm}
\begin{sloppypar}
For any discrete memoryless channel and decoding metric with ${\bar{C}_0=0}$,
\end{sloppypar}
\begin{equation}
\label{EUpper2}
E^q(0^+) \leq \max_{Q \in \mathcal{P}(\mathcal{X})} \sup_{s \geq 0} \sum_a \sum_b Q(a) Q(b) \hat{\mu}_{a,b}(s) \triangleq E^q_{\mathrm{up}}(0^+).
\end{equation}
\end{thm}

In such a case, the maximum distance between the expurgated lower bound and our upper bound on $E^q(0^+)$ can be estimated as follows:
\begin{align}
\big\lvert E^q_{\text{up}}&(0^+) - E^q_{\text{ex}}(0^+)\big\rvert \\
	&\leq \frac{1}{2}\max_{Q \in \mathcal{P}(\mathcal{X})} \sup_{s \geq 0} \sum_a \sum_b Q(a)Q(b) \notag\\
	&\hspace{3.5em}\big\lvert \hat{\mu}_{a,b}(s) + \hat{\mu}_{b,a}(s) - \mu_{a,b}(s) - \mu_{b,a}(s)\big)\big\rvert \\
	&= \frac{1}{2}\max_{Q \in \mathcal{P}(\mathcal{X})} \sup_{s \geq 0} \sum_{(a,b) \in \mathcal{B}} Q(a)Q(b) \notag\\
	&\hspace{3.5em}\big\lvert \hat{\mu}_{a,b}(s) + \hat{\mu}_{b,a}(s) - \mu_{a,b}(s) - \mu_{b,a}(s)\big)\big\rvert
\end{align}
\begin{align}
\hspace{1em}&= \frac{1}{2}\max_{Q \in \mathcal{P}(\mathcal{X})} \sum_{(a,b) \in \mathcal{B}} Q(a)Q(b) \notag\\
	&\hspace{3.5em}\big\lvert \hat{\mu}_{a,b}(0) + \hat{\mu}_{b,a}(0) - \mu_{a,b}(0) - \mu_{b,a}(0)\big)\big\rvert \\
	&\leq \frac{1}{2}\max_{(a,b) \in \mathcal{B}}\big\lvert \hat{\mu}_{a,b}(0) + \hat{\mu}_{b,a}(0) - \mu_{a,b}(0) - \mu_{b,a}(0)\big)\big\rvert \\
	&= \frac{1}{2}\max_{(a,b) \in \mathcal{B}} \Bigg(\log \frac{\sum_{y \in \hat{\mathcal{Y}}_{a,b}} W(y|a)}{\sum_{y \in \hat{\mathcal{Y}}_{a,b}^A} W(y|a)}\notag \\
	&\hspace{9.5em}+ \log \frac{\sum_{y \in \hat{\mathcal{Y}}_{a,b}} W(y|b)}{\sum_{y \in \hat{\mathcal{Y}}_{a,b}^A} W(y|b)}\Bigg). \label{NonBalancedErrBound}
\end{align}
Notice that for balanced pairs, definitions \eqref{BalancedCond} and \eqref{YASet} show that the sets $\hat{\mathcal{Y}}_{a,b}$ and $\hat{\mathcal{Y}}_{a,b}^A$ are equal, and therefore the quantity in \eqref{NonBalancedErrBound} is zero, as expected.
\begin{ex}
Consider the non-balanced channel-decoding metric pair of Example 1. For the pair of inputs $(0,1)$ one has $\hat{\mathcal{Y}}_{a,b} = \{0,1\} \neq \hat{\mathcal{Y}}_{a,b}^A = \{1\}$. Therefore, the upper bound on the gap $\big\lvert E^q_{\text{up}}(0^+) - E^q_{\text{ex}}(0^+)\big\rvert$ in equation \eqref{NonBalancedErrBound} for this channel-decoding pair is equal to
\begin{equation}
\big\lvert E^q_{\text{up}}(0^+) - E^q_{\text{ex}}(0^+)\big\rvert \leq \frac{1}{2}\bigg(\log \frac{1}{\varepsilon} + \log\frac{1-\varepsilon}{1-\varepsilon}\bigg) = \frac{1}{2}\log\frac{1}{\varepsilon}.
\end{equation}
\end{ex}

\begin{rem}
Converse bounds for codes at rate $R=0$ can often be used to also deduce bounds at $R>0$ through appropriate code coverings. Our bound, as most of the bounds on zero-rate codes, is based on the Plotkin double counting trick of Lemma \ref{PlotkinLemma}, which is used in \eqref{DMinUpperBound3}. In the same way as the Plotkin bound on the minimum Hamming distance of codes at $R=0$ can be extended to $R>0$ to deduce the Singleton bound and the Elias bound, our result can also be applied to derive bounds at $R>0$. For maximum likelihood decoding, the idea was initially presented by Blahut, although with a technical error in the proof which can however be corrected by means of the Ramsey-theoretic procedure also used here (see \cite{bondaschi1}). A similar extension can be derived for mismatched decoding. We do not expand on this point here since it would in large part repeat the discussion in \cite{bondaschi1} while at the same time requiring a significant technical digression on constant composition codes and the maximization over $s$, which would take us far from the main focus of this paper.
\end{rem}

%\appendices
%\section{Proof of the First Zonklar Equation}
%Appendix one text goes here.
%
%% you can choose not to have a title for an appendix
%% if you want by leaving the argument blank
%\section{}
%Appendix two text goes here.
%
%
%% use section* for acknowledgment
%\section*{Acknowledgment}
%
%
%The authors would like to thank...
%

% Can use something like this to put references on a page
% by themselves when using endfloat and the captionsoff option.
\ifCLASSOPTIONcaptionsoff
  \newpage
\fi

\begin{IEEEbiographynophoto}{Marco Bondaschi}
is a PhD student in the Laboratory for Information in Networked Systems at \'{E}cole Polytechnique F\'{e}d\'{e}rale de Lausanne (EPFL), Switzerland. He received the Bachelor's degree in Electronics Engineering and the Master's degree in Communication Sciences and Multimedia from the University of Brescia in 2017 and 2019 respectively. His main research interests are in information theory and learning theory.
\end{IEEEbiographynophoto}
\vfill
\newpage
\begin{IEEEbiographynophoto}{Albert Guill\'en i F\`abregas}
(S--01, M--05, SM--09, F--22) received the Telecommunications Engineering Degree and
the Electronics Engineering Degree from Universitat Polit\`ecnica de
Catalunya and Politecnico di Torino, respectively in 1999, and the Ph.D.
in Communication Systems from \'Ecole Polytechnique F\'ed\'erale de
Lausanne (EPFL) in 2004.

In 2020, he returned to a full-time faculty position at the Department
of Engineering, University of Cambridge, where he had been a full-time
faculty and Fellow of Trinity Hall from 2007 to 2012. Since 2011 he has
been an ICREA Research Professor at Universitat Pompeu Fabra (currently
on leave), where he is now an adjunct researcher. He has held appointments at the New Jersey Institute of
Technology, Telecom Italia, European Space Agency (ESA), Institut
Eur\'ecom, University of South Australia, Universitat Pompeu Fabra, University of Cambridge, as
well as visiting appointments at EPFL, \'Ecole Nationale des
T\'el\'ecommunications (Paris), Universitat Pompeu Fabra, University of
South Australia, Centrum Wiskunde \& Informatica and Texas A\&M
University in Qatar. His specific research interests are in the areas of
information theory, communication theory, coding theory, statistical inference.

Dr. Guill\'en i F\`abregas is a Member of the Young Academy of Europe,
and received the Starting and Consolidator Grants from the European
Research Council, the Young Authors Award of the 2004 European Signal
Processing Conference, the 2004 Best Doctoral Thesis Award from the
Spanish Institution of Telecommunications Engineers, and a Research
Fellowship of the Spanish Government to join ESA. Since 2013 he has been an Editor of the
Foundations and Trends in Communications and Information Theory, Now
Publishers and was an Associate Editor of the \sc{IEEE Transactions on
Information Theory} (2013--2020) and \sc{IEEE Transactions on Wireless
Communications} (2007--2011).
\end{IEEEbiographynophoto}
\begin{IEEEbiographynophoto}{Marco Dalai} 
(S'05-A'06-M'11-SM'17) received the degree in Electronic Engineering (cum laude) and the PhD in Information Engineering in 2003 and 2007 respectively from the University of Brescia Italy, where he is now an associate professor with the Department of Information Engineering. He is a member of the IEEE Information Theory Society, recipient of the 2014 IEEE Information Theory Society Paper Award and currently an Associate Editor of the IEEE Transactions on Information Theory.
\end{IEEEbiographynophoto}
\vfill
% that's all folks

\begin{thebibliography}{99}
\bibitem{csiszar1}
I. Csisz\'{a}r and P. Narayan, ``Channel capacity for a given decoding metric,'' \textit{IEEE Trans. Inf. Theory}, vol. 45, no. 1, pp. 35-43, 1995.
\bibitem{merhav1}
N. Merhav, G. Kaplan, A. Lapidoth, and S. Shamai, ``On information rates for mismatched decoders,'' \textit{IEEE Trans. Inf. Theory}, vol. 40, no. 6, pp. 1953–1967, Nov. 1994.
\bibitem{scarlett1}
J. Scarlett, A. Guillén i Fàbregas, A. Somekh-Baruch and A. Martinez, ``Information-theoretic foundations of mismatched decoding,'' \textit{Found. Trends Commun. Inf. Theory}, vol. 17, no. 2–3, pp. 149-401, 2020.
\bibitem{fischer1}
T. R. M. Fischer, ``Some remarks on the role of inaccuracy in Shannon’s theory of information transmission,'' \textit{Trans. 8th Prague Conf. Inf. Theory}, pp. 211–226, 1978.
\bibitem{kaplan1}
G. Kaplan and S. Shamai, ``Information rates and error exponents of compound channels with application to antipodal signaling in a fading environment,'' \textit{Arch. Elek. Uber.}, vol. 47, no. 4, pp. 228–239, 1993.
\bibitem{somekh1}
A. Somekh-Baruch, J. Scarlett, and A. Guillén i Fàbregas, ``Generalized random Gilbert-Varshamov codes,'' \textit{IEEE Trans. Inf. Theory}, vol. 65, no. 6, pp. 3452–3469, 2019.
\bibitem{kangarshahi1}
E. Asadi Kangarshahi and A. Guillén i Fàbregas, ``A single-letter upper bound to the mismatch capacity,'' \textit{IEEE Trans. Inf. Theory}, vol. 67, no. 4, pp. 2013-2033, 2021.
\bibitem{somekh2}
A. Somekh-Baruch, ``A single-letter upper bound on the mismatch capacity via a multicasting approach,'' \textit{Proc. 2020 IEEE Inf. Theory Workshop}, 2021. Available online: https://arxiv.org/pdf/2007.14322.pdf
\bibitem{kangarshahi2}
E. Asadi Kangarshahi and A. Guillén i Fàbregas, ``A sphere-packing exponent for mismatched decoding,'' \textit{Proc. 2021 IEEE Int. Symp. Inf. Theory}, to appear.
\bibitem{scarlett2}
J. Scarlett, L. Peng, N. Merhav, A. Martinez and A. Guillén i Fàbregas, ``Expurgated random-coding ensembles: exponents, refinements, and connections,'' \textit{IEEE Trans. Inf. Theory}, vol. 60, no. 8, pp. 4449-4462, 2014.
\bibitem{csiszar2}
I. Csisz\'{a}r and J. K\"{o}rner, \textit{Information theory: coding theorems for discrete memoryless systems}. Cambridge University Press, 2011.
\bibitem{cover1}
T. M. Cover and J. A. Thomas, \textit{Elements of information theory}. John Wiley \& Sons, 2012.
\bibitem{shannon1}
C. E. Shannon, ``The zero error capacity of a noisy channel,'' \textit{IEEE Trans. Inf. Theory}, vol. 2, no. 3, pp. 8-19, 1956.
\bibitem{komlos1}
J. Koml\'{o}s, ``A strange pigeon-hole principle,'' \textit{Order}, vol. 7, no. 2, pp. 107-113, 1990.
\bibitem{blinovsky1}
V. M. Blinovsky, ``New Approach to Estimation of the Decoding Error Probability,'' \textit{Prob. Inf. Trans.}, vol. 38, no. 1, pp. 16-19, 2002.
\bibitem{blinovsky2}
V. M. Blinovsky, ``Code bounds for multiple packings over a nonbinary finite alphabet,'' \textit{Prob. Inf. Trans.}, vol. 41, no. 1, pp. 23-32, 2005.
\bibitem{ahlswede1}
R. Ahlswede, V. Blinovsky, ``Multiple packing in sum-type metric spaces,'' \textit{Discrete Applied Mathematics}, vol. 156, no. 9, pp. 1469-1477, 2008.
\bibitem{diestel1}
R. Diestel, \textit{Graph Theory}, Springer-Verlag Berlin Heidelberg, 2017.
\bibitem{bondaschi1}
M. Bondaschi and M. Dalai, ``A Revisitation of Low-Rate Bounds on the Reliability Function of Discrete Memoryless Channels for List Decoding,'' \textit{IEEE Trans. Inf. Theory}, submitted.
\bibitem{sgb1}
C. E. Shannon, R. G. Gallager, and E. R. Berlekamp, ``Lower bounds to error probability for coding on discrete memoryless channels. II,'' \textit{Information and Control}, vol. 10, no. 5, pp. 522-552, 1967.
\end{thebibliography}
\end{document}